\documentclass[runningheads,orivec]{llncs}
\usepackage{graphicx}
\usepackage{comment}
\usepackage{url}
\usepackage[strings]{underscore}

\usepackage[utf8]{inputenc}
\usepackage[english]{babel}
\usepackage[T1]{fontenc}
\usepackage{textcomp}
\usepackage[titletoc,title]{appendix}

\usepackage{latexsym,amssymb,ae,amscd,stmaryrd}
\usepackage{graphicx}
\usepackage{enumerate}
\usepackage{listings} 
\usepackage{courier}
\usepackage{mathtools}
\usepackage{array} 
\usepackage{ebproof} 
\usepackage{xtab}

\allowdisplaybreaks
\usepackage{paralist}

\newcommand\sem[1]{\llbracket #1\rrbracket}
\newcommand\ofSig[1]{#1^\circ}
\newcommand\dom{\mathrm{dom}}
\newcommand\boldemph[1]{\textbf{\emph{#1}}}

\newcommand\cutout[1]{}

\newcommand{\ret}{ret} 
\newcommand{\nil}{\texttt{nil}} 
\newcommand{\fail}{\texttt{fail}} 
\newcommand{\bigchi}{\mbox{\Large$\chi$}} 
\newcommand{\refs}{\texttt{Refs}}
\newcommand{\meths}{\texttt{Meths}}
\newcommand{\vars}{\texttt{Vars}}

\newcommand{\vals}{\texttt{Vals}}
\newcommand{\cvals}{\texttt{CVals}}

\newcommand{\ssarefs}{\texttt{SSAVars}}
\newcommand{\termsM}{\texttt{Terms}}
\newcommand{\pts}{\texttt{Pts}}
\newcommand\atype[1]{\texttt{#1}}
\newcommand\Item[1][]{%
	\ifx\relax#1\relax  \item \else \item[#1] \fi
	\abovedisplayskip=0pt\abovedisplayshortskip=0pt~\vspace*{-\baselineskip}}
\newcommand{\preserves}{\simeq}
\newcolumntype{L}{>{$}l<{$}} 
\newcolumntype{R}{>{$}r<{$}} 
\newcolumntype{C}{>{$}c<{$}} 
\newcommand\TODO[1]{}

\newcommand\pair[1]{\langle #1 \rangle}
\newcommand{\CC}{C\nolinebreak\hspace{-.05em}\raisebox{.4ex}{\tiny\bf +}\nolinebreak\hspace{-.10em}\raisebox{.4ex}{\tiny\bf +}}

\lstdefinelanguage{Lambda}{%
  morekeywords={%
    if,then,else,fix,let,letrec,in,fail,nil,match,with,assert,new 
  },%
  morekeywords={[2]int},   
  otherkeywords={:,+,-,=}, 
  literate={
    {rarr}{{$\to$}}{1}
    {lambda}{{$\lambda$}}{1}
    {leq}{{$\leq$}}{1}
    {geq}{{$\geq$}}{1}
  },
  basicstyle={\small},
  keywordstyle={\bfseries},
  keywordstyle={[2]\itshape}, 
  keepspaces,
  tabsize=2,
  columns=flexible,
  mathescape 
}[keywords,comments,strings]%

\newcommand{\appsigma}{\{\sigma\}}
\newcommand{\appsigmas}[1]{\{\sigma_#1\}}

\newcommand\letin[1]{\text{let}~#1~\text{in}}
\newcommand{\tnil}{$\llbracket M,R,C,D,\phi,\nil\rrbracket = (\ret,(\ret=\nil)\land \phi,R,C,D)$}
\newcommand{\tfail}{$\llbracket \fail,R,C,D,\phi,k\rrbracket = (\ret,(\ret=\fail)\land \phi,R,C,D)$}
\newcommand{\tval}{$\begin{aligned}[t]
&\llbracket v,R,C,D,\phi,k\rrbracket = (\ret,(\ret=v)\land \phi,R,C,D)\\
&\qquad\qquad\qquad\quad\;\ \text{where $v = i,(),x,m$}
\end{aligned}$}
\newcommand{\tderef}{$\llbracket {!}r,R,C,D,\phi,k\rrbracket = (\ret,(\ret=D(r))\land \phi,R,C,D)$}
\newcommand{\tlambda}{$\begin{aligned}[t]
&\llbracket \lambda x.M,R,C,D,\phi,k\rrbracket = (\ret,(\ret=m)\land \phi,R',C,D)\\
&\qquad\qquad\qquad\quad\;\ \text{where $R'=R[m \mapsto \lambda x.M]$ and $m$ fresh}
\end{aligned}$}
\newcommand{\tpi}{$
	\begin{aligned}[t]
		& \llbracket\pi_i\,M,R,C,D,\phi,k\rrbracket =\\
		&\quad \begin{aligned}[t]
			& \letin{(\ret_1,\phi_1,R_1,C_1,D_1) = \llbracket M,R,C,D,\phi,k\rrbracket}\\
			& (\ret,(F~\ret_1~\ret~(\ret=\pi_i\,ret_1))\land \phi_1,R_1,C_1,D_1)
		\end{aligned}
	\end{aligned}$}
\newcommand{\tassign}{$
	\begin{aligned}[t]
		& \llbracket r:=M,R,C,D,\phi,k\rrbracket =\\
		&\quad \begin{aligned}[t]
			& \letin{(\ret_1,\phi_1,R_1,C_1,D_1) = \llbracket M,R,C,D,\phi,k\rrbracket}\\
			& \letin{C_1' = C_1[r]}\
			\letin{D_1' = D_1[r\mapsto C_1'(r)]}\\
			& (\ret,(F~\ret_1~\ret~((\ret={()}) \land (D_1'(r)=\ret_1)))\\
			& \qquad {}\land \phi_1,R_1,C_1',D_1')
		\end{aligned}
	\end{aligned}$}
\newcommand{\tplus}{$
	\begin{aligned}[t]
		& \llbracket M_1 \oplus M_2,R,C,D,\phi,k\rrbracket =\\
		&\quad \begin{aligned}[t]
			& \letin{(\ret_1,\phi_1,R_1,C_1,D_1) = \llbracket M_1,R,C,D,\phi,k\rrbracket}\\
			& \letin{(\ret_2,\phi_2,R_2,C_2,D_2) = 
			 \llbracket M_2,R_1,C_1,D_1,\phi_1,k\rrbracket}\\
			& (\ret,(F~\ret_1~\ret~(F~\ret_2~\ret~(\ret=\ret_1 \oplus \ret_2)))\\
			&\qquad \land \phi_2,R_2,C_2,D_2)
		\end{aligned}
	\end{aligned}$}
\newcommand{\tpair}{$
	\begin{aligned}[t]
	& \llbracket \pair{M_1,M_2},R,C,D,\phi,k\rrbracket =\\
		&\quad \begin{aligned}[t]
			& \letin{(\ret_1,\phi_1,R_1,C_1,D_1) = \llbracket M_1,R,C,D,\phi,k\rrbracket}\\
			& \letin{(\ret_2,\phi_2,R_2,C_2,D_2) =
			 \llbracket M_2,R_1,C_1,D_1,\phi_1,k\rrbracket}\\
			& (\ret,(F~\ret_1~\ret~(F~\ret_2~\ret~(\ret= \pair{\ret_1,\ret_2})))\\
			&\qquad \land \phi_2,R_2,C_2,D_2)
		\end{aligned}
	\end{aligned}$}
\newcommand{\tletin}{$
	\begin{aligned}[t]
	& \llbracket \texttt{let $x=M$ in $M'$},R,C,D,\phi,k\rrbracket =\\
		&\quad \begin{aligned}[t]
			& \letin{(\ret_1,\phi_1,R_1,C_1,D_1) = \llbracket M,R,C,D,\phi,k\rrbracket}\\
			& \letin{(\ret_2,\phi_2,R_2,C_2,D_2) =\\ 
			& \qquad \llbracket M'\{\ret_1/x\},R_1,C_1,D_1,\phi_1,k\rrbracket}\\
			& (\ret,(F~\ret_1~\ret~(F~\ret_2~\ret~(\ret=\ret_2)))\\
			& \qquad {}\land \phi_2,R_2,C_2,D_2)
		\end{aligned}
	\end{aligned}$}
\newcommand{\tletrec}{$
	\begin{aligned}[t]
	& \llbracket \texttt{letrec $f=\lambda x.M$ in $M'$},R,C,D,\phi,k\rrbracket =\\
		&\quad \begin{aligned}[t]
			& \letin{\text{$m,f'$ be fresh}}\\
            & \letin{R' = R[m \mapsto \lambda x.M\{f'/f\}]}\\
			& \llbracket M'\{f'/f\},R',C,D,\phi \land (f' = m),k\rrbracket
		\end{aligned}
	\end{aligned}$}
\newcommand{\tapplym}{$
	\begin{aligned}[t]
		& \llbracket m\,M,R,C,D,\phi,k\rrbracket =\\
		&\quad \begin{aligned}[t]
			& \letin{(\ret_1,\phi_1,R_1,C_1,D_1) = \llbracket M,R,C,D,\phi,k\rrbracket}\\
			& \letin{R(m)\text{ be }\lambda x.N}\\
			& \letin{(\ret_2,\phi_2,R_2,C_2,D_2) = \\
			& \qquad \llbracket N\{\ret_1/x\},R_1,C_1,D_1,\phi_1,k-1\rrbracket}\\
			& (\ret,(F~\ret_1~\ret~(F~\ret_2~\ret~(\ret=\ret_2)))\\
			& \qquad \land \phi_2,R_2,C_2,D_2)
		\end{aligned}
	\end{aligned}$}
\newcommand{\tifthenelse}{$
	\begin{aligned}[t]
		& \llbracket \texttt{if $M_b$ then $M_1$ else $M_0$},R,C,D,\phi,k\rrbracket =\\
		&\quad \begin{aligned}[t]
			& \letin{(\ret_b,\phi_b,R_b,C_b,D_b) = \llbracket M_b,R,C,D,\phi,k\rrbracket}\\
			& \letin{(\ret_0,\phi_0,R_0,C_0,D_0) = 
			 \llbracket 	M_0,R_b,C_b,D_b,\phi_b,k\rrbracket}\\
			& \letin{(\ret_1,\phi_1,R_1,C_1,D_1) = 
			 \llbracket M_1,R_0,C_0,D_b,\phi_0,k\rrbracket}\\
			& \letin{C' = C_1[r_1]\cdots[r_n]\ (\varPi=\{r_1,\dots,r_n\})}\\
			& \letin{\psi_0 = (\ret_b = 0) \implies (F~\ret_0~\ret~((\ret=\ret_0) \\
                          &\;\,\qquad\qquad\qquad\qquad\qquad
                          {}\land\!\! \bigwedge_{r \in \varPi}\!\! (C'(r)=D_1(r))))}\\
                      & \letin{\psi_1 = (\ret_b \neq 0) \implies
                        (F~\ret_1~\ret~((\ret=\ret_1) \\
                         &\;\,\qquad\qquad\qquad\qquad\qquad
                        {}\land\!\!\bigwedge_{r \in \varPi}\!\! (C'(r)=D_1(r))))}\\
			& (\ret,(F~\ret_b~\ret~(\psi_0 \land \psi_1))\land \phi_1,R_1,C',C')
		\end{aligned}
	\end{aligned}$}
\newcommand{\tapplyx}{$
	\begin{aligned}[t]
		& \llbracket x^\theta\,M,R,C,D,\phi,k\rrbracket =\\
		&\quad \begin{aligned}[t]
                  & \letin{(\ret_0,\phi_0,R_0,C_0,D_0) = \llbracket M,R,C,D,\phi,k\rrbracket}\\
                  &\text{if }R \upharpoonright \theta =\emptyset\text{ then }(\ret,(\ret=\nil)\land\phi_0,R_0,C_0,D_0)\\
&\text{else }
			\letin{R \upharpoonright \theta\text{ be }\{m_1,...,m_n\}}\\
			&\text{for each }i \in \{1,...,n\}:\\
			&\quad
			\begin{aligned}[t]
				& \letin{R(m_i)\text{ be } \lambda y_i.N}\\
				& \letin{(\ret_i,\phi_i,R_i,C_i,D_i) = \\
				& \qquad \llbracket N_i\{\ret_0/y_i\},R_{i-1},
				C_{i-1},D_0,\phi_{i-1},k-1\rrbracket}\\
			\end{aligned}\\
			& \letin{C_n' = C_n[r_1]\cdots[r_j]\ (\varPi=\{r_1,\dots,r_j\})}\\
			& \letin{\psi = \bigwedge_{i=1}^{n}
				\left(\vcenter{\hbox{$\displaystyle
				\begin{aligned}[t]
					&(x=m_i) \implies \\
					&((F~\ret_i~\ret~(\ret=\ret_i)) \land\\ 
					&\bigwedge_{r \in \varPi}{(C_n'(r)=D_i(r))})
				\end{aligned}
				$}}\right)}\\
			& (\ret,(F~\ret_0~\ret~\psi)\land \phi_n,R_n,C_n',C_n')
			\end{aligned}
	\end{aligned}$}

\begin{document}
\title{Higher-Order Bounded Model Checking}
%
%
\author{Yu-Yang Lin \and Nikos Tzevelekos}
%
%
\institute{Queen Mary University of London}
\maketitle              
\begin{abstract}
  We present a Bounded Model Checking technique for higher-order programs. The vehicle of our study is a higher-order calculus with general references. Our technique is a symbolic state syntactical translation based on SMT solvers, adapted to a setting where the values passed and stored during computation can be functions of arbitrary order. We prove that our algorithm is sound, and devise an optimisation based on points-to analysis to improve scalability. We moreover provide a prototype implementation of the algorithm with experimental results showcasing its performance.

\end{abstract}
\section{Introduction}

Bounded Model Checking~\cite{DBLP:conf/tacas/BiereCCZ99} (BMC) is a
model checking technique that allows for
highly automated and scalable SAT/SMT-based verification and has been widely used to find errors in C-like languages~\cite{DBLP:conf/tacas/ClarkeKL04,DBLP:conf/tacas/MorseRCN014,DBLP:journals/tcad/DSilvaKW08,DBLP:conf/tacas/AmlaKMM03}.
BMC amounts to bounding the executions of programs by unfolding loops only up to a given bound, and model checking the resulting execution graph. Since the advent of~\textsc{Cbmc}~\cite{DBLP:conf/tacas/ClarkeKL04}, the mainstream approach additionally proceeds by symbolically executing program paths and gathering the resulting path conditions in propositional formulas which can then be passed on to  SAT/SMT solvers.
Thus, BMC performs a syntactic translation of program source code into a propositional formula, and uses the power of SAT/SMT solvers to check the bounded behaviour of programs.

Being a Model Checking technique, BMC has the ability to produce \textit{counterexamples}, which are execution traces that lead to the violation of desired properties. A specific advantage of BMC over unbounded techniques is that it avoids the full effect of state-space explosion at the expense of full verification. On the other hand, since BMC is inconclusive if the formula is unsatisfiable, it is generally regarded as a bug-finding or underapproximation technique, which lets it avoid spurious errors. While it tends to be the most empirically effective approach for ``shallow" bugs~\cite{DBLP:journals/tcad/DSilvaKW08,DBLP:conf/tacas/AmlaKMM03}, bugs in deep loops and recursion are often a weakness. It is only possible to prove complete correctness if bounds for loops and recursion are determinable.

The above approach has been predominantly applied to imperative, first-order languages and, while tools like
\textsc{Cbmc} can handle {\CC} (and, more recently, \textsc{Java} bytecode), the foundations of BMC for higher-order programs have not been properly laid. This is what we address herein.
%
%
%
We propose a symbolic BMC procedure for higher-order functional/imperative programs that may contain free variables of ground type. Our contributions involve a syntactical translation to apply BMC to higher-order languages with higher-order state, a proof that the approach is sound, an optimisation based on points-to analysis to improve scalability, and a prototype implementation of the procedure with experimental results showcasing its performance.

As with most approaches to software BMC, we translate a given higher-order program into a propositional formula for an SMT solver to check for satisfiability, where formulas are satisfiable only if a violation is reachable within a given bound.
Where in first-order programs BMC places a bound on loop unfolding, in the higher-order setting we place the bound on nested recursive calls. 
The main challenge for the translation then is the symbolic execution of paths which involve the flow of higher-order terms, by either variable binding or use of the store. This is achieved by adapting the standard technique of Static Single Assignment to a setting where variables/references can be of  higher order.
To handle higher-order terms in particular, we use a nominal approach to methods, whereby each method is uniquely identified by a name. We capture program behaviour by also uniquely identifying every step in the computation tree with a return variable; analogous to how \textsc{Cbmc}~\cite{DBLP:conf/tacas/ClarkeKL04} captures the behaviour of sequencing commands in ANSI-C programs.

To give a simple example of the approach, consider the following code, where $r$ is a reference of type $\texttt{int}\to\texttt{int}$, and $f,g,h$ are variables of 
type $\texttt{int}\to\texttt{int}$, and $n,x$ are variables of type $\texttt{int}$.
\begin{lstlisting}[numbers=left,language=lambda,xleftmargin=1em]
let f = lambda x,g,h. if (x <= 0) then g else h 
in              
r := f n (lambda x. x-1) (lambda x. x+1);
assert(!r n >= n)
\end{lstlisting}
In the code above, a function is assigned to reference $r$. In a symbolic setting, it is not immediately obvious which function to call when dereferencing $r$ in line~4.
Luckily, we know that when calling $f$ in line~3, its value can only be the one bound to it in line~1. Thus, a first transformation of the code could be:
\begin{lstlisting}[numbers=left,language=lambda,xleftmargin=1em,firstnumber=3]
r := if (n <= 0) then (lambda x. x-1) else (lambda x. x+1);
assert(!r n >= n)
\end{lstlisting}
The assignment in line~3 can be facilitated by using a return variable $\ret$ and method names for $(\lambda x . x-1)$ and $(\lambda x . x + 1)$:
\begin{lstlisting}[numbers=left,language=lambda,xleftmargin=1em]
let m1 = lambda x. x-1 in let m2 = lambda x. x+1 in 
let ret = if (n <= 0) then m1 else m2 in
r := ret;
assert(!r n >= n)
\end{lstlisting}
We now need to symbolically decide how to dereference $r$. 
The simplest solution is try to match $r$ with all existing functions of matching type, in this case $m1$ and $m2$:
\begin{lstlisting}[numbers=left,language=lambda,xleftmargin=1em]
let m1 = lambda x. x-1 in let m2 = lambda x. x+1 in 
let ret = if (n <= 0) then m1 else m2 in
r := ret;
let ret' = match r with
           | m1 -> m1 n
           | m2 -> m2 n in
assert(ret' >= n)
\end{lstlisting}
Performing the substitutions of $m1,m2$, we can read off the following formula for checking falsity of the assertion:
\[\begin{tabular}{@{}R@{}L}
    (\ret'<n) &{}\land (r=m1 \Rightarrow ret' = n-1)\\
              &{}\land (r=m2 \Rightarrow ret' = n+1) \land (r=ret)\\
    &{}\land(n<=0\Rightarrow ret=m1)\land(n>0\Rightarrow ret=m2)
\end{tabular}\]
The above is true e.g.\ for $n=0$, and hence the code violates the assertion.

These ideas underpin our first BMC translation, which is presented in Section~\ref{sec:trans}.
The language we examine, \textsc{HORef}, is a higher-order language with general references and integer arithmetic.
While correct,
one can quickly see that our first translation is inefficient when trying to resolve the flow of functions to references and variables. In effect, it explores all possible methods of the appropriate type that have been created so far, and relies on the solver to pick the right one. In Section~\ref{sec:points-to}, we optimise the translation by  restricting such choices according to an analysis akin to \emph{points-to analysis}~\cite{DBLP:conf/popl/Steensgaard96,Andersen:94:PhD}.
Finally, in Section~\ref{sec:tool} we present an implementation of our technique in a BMC tool for a higher-order \textsc{OCaml}-like syntax extending \textsc{HORef} and test it on several example programs adapted from the \textsc{MoCHi} benchmark~\cite{DBLP:conf/pepm/SatoUK13}.

\section{The Language: \textsc{HORef}}

\newcommand\preG{} 
\newcommand\vphi[1]{}

\begin{figure*}[t]
	\centering
	\setlength{\tabcolsep}{2pt}
	\renewcommand{\arraystretch}{1.5}
	\begin{tabular}{R R L}
		\textit{Terms}& \quad \termsM \ \ni \ M ::=&\fail \mid x \mid m \mid i \mid ()
		\mid r:=M 
		\mid {!r} 
		\mid M \oplus M 
		\mid \langle M,M \rangle \\
		&&\mid \pi_1\,M 
		\mid \pi_2\,M 		
		\mid x~M 
		\mid m~M 
		\mid \texttt{if $M$ then $M$ else $M$} \\
		&&\mid  \texttt{let $x = M$ in $M$}
        \mid  \texttt{letrec $x = \lambda x.M$ in $M$}
		\mid  \lambda x . M\\
		\textit{Values}& \quad \vals \ \ni \ v ::=&x \mid m \mid i \mid () \mid \pair{v,v}
	\end{tabular}
	\renewcommand{\arraystretch}{3}
	\begin{tabular}{ c }
        \hline
{		\begin{prooftree}
			\hypo{ \vphi{()}  }
			\infer1[]{ \preG \fail : \theta }
		\end{prooftree}
}		
{		\begin{prooftree}
			\hypo{ \vphi{()} }
			\infer1[]{ \preG () : \atype{unit} }
		\end{prooftree}
}		
{		\begin{prooftree}
			\hypo{ \vphi{()} }
			\infer1[]{ \preG i : \atype{int} }
		\end{prooftree}
}		
{		\begin{prooftree}
			\hypo{ x\in\vars_\theta }
			\infer1[]{ \preG x : \theta }
		\end{prooftree}
}		
{		\begin{prooftree}
			\hypo{ m \in \meths_{\theta,\theta'} }
			\infer1[]{ \preG m : \theta \rightarrow \theta' }
		\end{prooftree}
}		
{		\begin{prooftree}
			\hypo{  M_1, M_2 : \atype{int} }
			\infer1[]{  M_1 \oplus M_2 : \atype{int} }
		\end{prooftree}}
		\\		
{		\begin{prooftree}
			\hypo{ \preG M : \atype{int} }
			\hypo{ \preG M_0 : \theta }
                        \hypo{ \preG M_1 : \theta }
			\infer3[]{ \preG \texttt{if $M$ then $M_1$ else $M_0$} : \theta }
		\end{prooftree}
}		
{		\begin{prooftree}
			\hypo{ \preG M_1 : \theta_1 }
			\hypo{ \preG M_2 : \theta_2 }
			\infer2[]{ \preG \langle M_1 , M_2 \rangle : \theta_1 \times \theta_2 }
		\end{prooftree}
}		
{		\begin{prooftree}
			\hypo{ \preG \langle M_1 , M_2 \rangle : \theta_1 \times \theta_2 }
			\infer1[]{ \preG \pi_i \langle M_1 , M_2 \rangle : \theta_i~~(i=1,2) }
		\end{prooftree}
}		\\
{		\begin{prooftree}
			\hypo{ r \in \refs_{\theta} }
			\infer1[]{ \preG {!r} : \theta }
		\end{prooftree}
}
{		\begin{prooftree}
			\hypo{ r \in \refs_{\theta} }
			\hypo{ \preG M : \theta }
			\infer2[]{ \preG r := M : \atype{unit} }
		\end{prooftree}
}
          \cutout{
		\begin{prooftree}
			\hypo{ r \in \refs_{\atype{int}} }
			\infer1[]{ \preG {!r} : \atype{int} }
		\end{prooftree}
		
		\begin{prooftree}
			\hypo{ r \in \refs_{\atype{int}} }
			\hypo{ \preG M : \atype{int} }
			\infer2[]{ \preG r := M : \atype{unit} }
		\end{prooftree}
		
		\begin{prooftree}
			\hypo{ r \in \refs_{\theta,\theta'} }
			\infer1[]{ \preG {!r} : \theta \rightarrow \theta' }
		\end{prooftree}
		
		\begin{prooftree}
			\hypo{ r \in \refs_{\theta,\theta'} }
			\hypo{ \preG M : \theta \rightarrow \theta' }
			\infer2[]{ \preG r := M : \atype{unit} }
		\end{prooftree}
          }	
{		\begin{prooftree}
                  \hypo{ M : \theta' }
                  \hypo{ x : \theta }
			\infer2[]{ \preG \lambda x . M : \theta \rightarrow \theta' }
		\end{prooftree}
}		
{		\begin{prooftree}
			\hypo{ \preG M : \theta }
			\hypo{  N : \theta' }
			\hypo{ x : \theta }
			\infer3[]{ \preG \texttt{let $x = M$ in $N$} : \theta' }
		\end{prooftree}
 }       \\
{		\begin{prooftree}
			\hypo{ M : \theta' }
			\hypo{  N : \theta'' }
			\hypo{ x : \theta \to \theta' }
            \hypo{ y : \theta }
			\infer4[]{ \preG \texttt{letrec $x = \lambda y . M$ in $N$} : \theta'' }
		\end{prooftree}
}		
{		\begin{prooftree}
			\hypo{ x: \theta \rightarrow \theta' }
			\hypo{ \preG M : \theta }
			\infer2[]{ \preG x~M : \theta' }
		\end{prooftree}
}		
{		\begin{prooftree}
			\hypo{ m \in \meths_{\theta \rightarrow \theta'} }
			\hypo{ \preG M : \theta }
			\infer2[]{ \preG m~M : \theta' }
		\end{prooftree}
}		
	\end{tabular}
	\caption{Grammar and typing rules for \textsc{HORef}}
	\label{fig:terms}	
\end{figure*}

Here we present a higher-order language with (higher-order) state, which we call \textsc{HORef}, as an idealised setting for languages like \textsc{Java} and \textsc{OCaml}. The syntax consists of a call-by-value $\lambda$-calculus with types
\[
\theta\ ::=\ \atype{unit}\mid\atype{int}\mid\theta\times\theta\mid\theta\to\theta
  \]
and references of arbitrary types.
\cutout{
we write
\begin{itemize}
 \item $x^\theta$ for a variable $x$ of type $\theta$; 
 \item $i$ for some integer; 
 \item $m$ for some method name; 
 \item $()$ for the unit or skip command;
 \item $\oplus$ for any binary operation on integers;
 \item $\pi_1$ and $\pi_2$ for the left and right projection of a pair;
 \item and $r$ for some reference. 
 \end{itemize}}%
We assume countable disjoint sets \vars, \refs\ and \meths, for \emph{variables}, \emph{references} and \emph{methods} respectively.
Variables are ranged over by $x$ and variants; references by $r$ and variants; and methods by $m$ and variants.
We assume these sets are typed, that is:
\begin{align*}
&\vars=\biguplus\nolimits_\theta\vars_\theta, \
\refs=\biguplus\nolimits_\theta\refs_\theta, \\
&\meths=\biguplus\nolimits_{\theta,\theta'}\meths_{\theta\to\theta'}.
\end{align*}
The syntax and typing rules are given in Figure~\ref{fig:terms}.
Note that we assume a set of arithmetic operators $\oplus$, which we leave unspecified as they do not affect the analysis.
We extend the syntax with usual constructs for sequencing and assertions: $M;N$ stands for $\texttt{let $\_=M$ {in} $N$}$;
while $r{++}$ is $r:={!r}+1$;
and $\texttt{assert}(M)$ is $\texttt{if $M$ then () else $\fail$}$ (with boolean values represented by $0,1$).

Note that the use of typed variables allows us to type terms without need for typing contexts. 
As usual, a variable occurrence is \emph{free} if it is not in the scope of a
matching
($\lambda$/$\texttt{let}$/$\texttt{letrec}$)-binder. Terms are considered modulo $\alpha$-equivalence and in particular we may assume that no variable occurs both as free and bound in the same term.
We call a term \emph{closed} if it contains no free variables.

References in our language are global, and there is no fresh reference creation construct (this choice made for simplicity). On the other hand, methods are dynamically created in terms, and for that reason we will be frequently referring to them as \boldemph{names}. The terminology comes from nominal techniques~\cite{nom1,nom2}.
On a related note, $\lambda$-abstractions are not values in our language. This is due to the fact that in the semantics these get evaluated to method names.

\cutout{
We do not have arbitrary term to term application. Method application in our syntax is restricted to names and variables applied to terms. This maintains expressiveness, however, as one can use let-bindings to chain method application. We also see in Figure \ref{fig:terms} the typing rules for the terms. }

In our approach, checking for violations of safety properties is reduced to the reachability of failure. We have therefore included the $\fail$ primitive for when a program
reaches a failure.
Accordingly, our bounded model checking routine will return $\fail$ when it 
aborts on reaching a failure and $\nil$ when it aborts on reaching the bound.
The use of $\nil$ is analogous to the {unwinding assertions} used in \textsc{Cbmc}. It is not part of the syntax of \textsc{HORef}.

\cutout{
Note that we do not have references themselves as values in our syntax. Instead, we only have handles to references. As such, creating new references is not possible, i.e. we do not allow syntax such as \texttt{let r = ref 0}. Instead, the size of the store is fixed, since all references used must exist. In practice, this simply means that all references must be declared beforehand.}

\paragraph{Bounded Operational Semantics}

We next present a bounded operational semantics for \textsc{HORef}, which is the one captured by our bounded BMC routine.
The semantics is parameterised by a bound $k$ which, similarly to loop unwinding in procedural languages, it bounds the depth of method (i.e.\ function) calls within an execution. A bound $k=0$
in particular means that, unless no method calls are made, execution will terminate returning $\nil$. 
Consequently, in this bounded operational semantics, all programs must halt; either when the program itself halts (returning a value or $\fail$), or when the bound is exceeded.
Note at this point that the standard (unbounded) semantics of \textsc{HORef}, allowing arbitrary recursion, can be obtained e.g.\ by allowing bound values $k=\infty$. 

To describe this behaviour, we chose a big-step operational semantics representation with rules of the form
\[(M,R,S,k)\Downarrow(\bigchi,R',S')\]
where
$\bigchi \in (\vals \cup \{\fail,\nil\})$. In other words, all terms must eventually evaluate to a value, $\fail$ or $\nil$.
A \boldemph{configuration} is a quadruple $(M,R,S,k)$ where $M$ is a typed term and:
\begin{itemize}
  \item $R:\meths\rightharpoonup\termsM$ is a finite map, called a \boldemph{method repository}, such that for all $m\in\dom(R)$, if $m\in\meths_{\theta\to\theta'}$ then $R(m)=\lambda x.M:\theta\to\theta'$.
  \item $S:\refs\rightharpoonup\vals$ is a finite map, called a \boldemph{store}, such that
    for all $r\in\dom(S)$, if $r\in\refs_{\theta}$ then $S(r):\theta$.
    \item $k\in\{\nil\}\cup\mathbb{N}$ is the nested calling bound, where decrementing $k$ beyond zero results in $\nil$.
    \end{itemize}
    A closed configuration is one all of whose components are closed.
    We call a configuration $(M,R,S,k)$ \boldemph{valid} if
    all methods and references appearing in $M,R,S$ are included in $\dom(R)$ and $\dom(S)$ respectively. A closed configuration is one all of whose components are closed.

\begin{definition}
The operational semantics is defined on closed valid configurations, by the rules given in Figure~\ref{fig:semantics}.
\end{definition}

\begin{figure*}
	\centering
	\setlength{\tabcolsep}{2pt}
	\renewcommand{\arraystretch}{1.5}
	\begin{tabular}{RL@{\qquad}RL@{\qquad}RL}
		(\Downarrow_\nil)& (M,R,S,\nil)\Downarrow(\nil,R,S)	
		&
		(\Downarrow_\text{\fail}) &(\fail,R,S,k)\Downarrow(\fail,R,S)\\
		(\Downarrow_\text{val}) & (v,R,S,k)\Downarrow(v,R,S)
		&
		(\Downarrow_\text{drf}) & ({!r},R,S,k)\Downarrow(S(r),R,S)\\
          (\Downarrow_\lambda)& 
                    \multicolumn{3}{L}{(\lambda x.M,R,S)\Downarrow(m,R[m\mapsto \lambda x.M],S)\text{ where $m\notin\dom(R)$}}
        \end{tabular}
        \renewcommand{\arraystretch}{3}
        \begin{tabular}{RLRL}
    	(\Downarrow_{\pi_i})&
    	{\begin{prooftree}
    	\hypo{ (M,R,S,k)\Downarrow(\pair{v_1,v_2},R',S') }
    	\infer1[]{ (\pi_i\,M,R,S,k)\Downarrow(v_i,R',S') }
    	\end{prooftree}}\text{$i = 1,2$}
        \quad
	(\Downarrow_{:=}) ~
    	{\begin{prooftree}
    	\hypo{ (M,R,S,k)\Downarrow(v,R',S') }
    	\infer1[]{ (r:=M,R,S,k)\Downarrow((),R',S'[r \mapsto v]) }
    	\end{prooftree}}
    	\\
    	(\Downarrow_\oplus)&
    	{\begin{prooftree}
    	\hypo{ (M_1,R,S,k)\Downarrow(i_1,R_1,S_1) &  (M_2,R_1,S_1,k)\Downarrow(i_2,R_2,S_2) }
    	\infer1[]{ (M_1 \oplus M_2,R,S,k)\Downarrow(i,R_2,S_2)\quad (\text{where }i=i_1 \oplus i_2) }
    	\end{prooftree}}
    	\\
    	(\Downarrow_\times)&
    	{\begin{prooftree}
    	\hypo{ (M_1,R,S,k)\Downarrow(v_1,R_1,S_1) & (M_2,R_1,S_1,k)\Downarrow(v_2,R_2,S_2) }
    	\infer1[]{ (\pair{M_1,M_2},R,S,k)\Downarrow(\pair{v_1,v_2},R_2,S_2) }
      \end{prooftree}}
    	\\ 
    	(\Downarrow_\text{let})&
    	{\begin{prooftree}
   		\hypo{ (M,R,S,k)\Downarrow(v,R',S') & (N\{v/x\},R',S',k)\Downarrow(v',R'',S'') }
   		\infer1[]
   		{ (\texttt{let $x=M$ in $N$},R,S,k)\Downarrow(v',R'',S'') }
              \end{prooftree}}
    	\\
    	(\Downarrow_{@})&
    	{\begin{prooftree}
		\hypo{ (M,R,S,k)\Downarrow(v,R',S') \;\; (N\{v/x\},R',S',k{-}1)\Downarrow(v',R'',S'') }
		\infer1[]{ (m\,M,R,S,k)\Downarrow(v',R'',S'') \quad(\text{where } R(m)=\lambda x.N)}
              \end{prooftree}}
   		\\ 
   		(\Downarrow_{\text{if}})&
        {\begin{prooftree}
        \hypo{ (M,R,S,k)\Downarrow(i,R',S') & (M_j,R',S',k)\Downarrow(v_j,R_j,S_j) }
        \infer1[]{ (\texttt{if $M$ then $M_1$ else $M_0$},R,S,k)\Downarrow(v_j,R_j,S_j)}
        \end{prooftree}}
           \\
          (\Downarrow_\text{rec})&
    	{\begin{prooftree}
   		\hypo{ (N[m/f],R[m\mapsto\lambda x.M\{m/f\}],S,k)\Downarrow(v,R',S') }
   		\infer1[]
   		{ (\texttt{letrec $f=\lambda x.M$ in $N$},R,S,k)\Downarrow(v',R',S') }
    	\end{prooftree}}
\\
    	(\not\Downarrow_{\pi_i})&
    	{\begin{prooftree}
    	\hypo{ (M,R,S,k)\Downarrow(\bigchi,R',S') }
    	\infer1[]{ (\pi_i\,M,R,S,k)\Downarrow(\bigchi,R',S') }
    	\end{prooftree}}
    	\text{$i = 1,2$}
        \quad
    	(\not\Downarrow_{:=})~
    	{\begin{prooftree}
    	\hypo{ (M,R,S,k)\Downarrow(\bigchi,R',S') }
    	\infer1[]{ (r:=M,R,S,k)\Downarrow(\bigchi,R',S') }
    	\end{prooftree}}
        \\
        (\not\Downarrow_{\oplus_1})&
    	{\begin{prooftree}
        \hypo{ (M_1,R,S,k)\Downarrow(\bigchi,R_1,S_1) }
    	\infer1[]{ (M_1 \oplus M_2,R,S,k)\Downarrow(\bigchi,R_1,S_1) }
        \end{prooftree}}
        \quad
        (\not\Downarrow_{\oplus_2})~
    	{\begin{prooftree}
        \hypo{ (M_1,R,S,k)\Downarrow(i_1,R_1,S_1) & (M_2,R_1,S_1,k)\Downarrow(\bigchi,R_2,S_2) }
    	\infer1[]{ (M_1 \oplus M_2,R,S,k)\Downarrow(\bigchi,R_2,S_2) }
        \end{prooftree}}
        \end{tabular}
    	\begin{tabular}{RLRL}
    	(\not\Downarrow_{\times1})&
    	{\begin{prooftree}
    	\hypo{ (M_1,R,S,k)\Downarrow(\bigchi,R_1,S_1) }
    	\infer1[]{ (\pair{M_1,M_2},R,S,k)\Downarrow(\bigchi,R_1,S_1) }
    	\end{prooftree}}            
        \quad
    	(\not\Downarrow_{\times2})~
        {\begin{prooftree}
    	\hypo{ (M_1,R,S,k)\Downarrow(v_1,R_1,S_1) & (M_2,R_1,S_1,k)\Downarrow(\bigchi,R_2,S_2) }
    	\infer1[]{ (\pair{M_1,M_2},R,S,k)\Downarrow(\bigchi,R_2,S_2) }
      \end{prooftree}}
\\
    	(\not\Downarrow_{\text{let}_1})&
    	{\begin{prooftree}
    	\hypo{ (M,R,S,k)\Downarrow(\bigchi,R',S') }
    	\infer1[]{ (\texttt{let $x=M$ in $M'$},R,S,k)\Downarrow(\bigchi,R',S') }
    	\end{prooftree}}
\\
              	(\not\Downarrow_{\text{let}_2})
&    	{\begin{prooftree}
    	\hypo{ (M,R,S,k)\Downarrow(v',R',S') & (N\{v/x\},R',S',k)\Downarrow(\bigchi,R'',S'') }
    	\infer1[]{ (\texttt{let $x=M$ in $M'$},R,S,k)\Downarrow(\bigchi,R'',S'') }
    	\end{prooftree}}
\\
    	(\not\Downarrow_{@_1})& 
    	{\begin{prooftree}
    	\hypo{ (M,R,S,k)\Downarrow(\bigchi,R',S') }
    	\infer1[]{ (m\,M,R,S,k)\Downarrow(\bigchi,R',S')}
    	\end{prooftree}}
\quad
    	(\not\Downarrow_{@_2})~
    	{\begin{prooftree}
    	\hypo{(M,R,S,k)\Downarrow(v,R',S') \;\; (N\{v/x\},R',S',k{-}1)\Downarrow(\bigchi,R'',S'') }
    	\infer1[]{ (m\,M,R,S,k)\Downarrow(\bigchi,R'',S'') \quad
        (\text{where $R(m)=\lambda x.N$})}
    	\end{prooftree}}
\\
(\not\Downarrow_{\text{if}_1})
    	&{\begin{prooftree}
    	\hypo{ (M,R,S,k)\Downarrow(\bigchi,R',S') }
    	\infer1[]{ (\texttt{if $M$ then $M_1$ else $M_0$},R,S,k)\Downarrow(\bigchi,R',S')}
    	\end{prooftree}}
\\
        	(\not\Downarrow_{\text{if}_2})&
            {\begin{prooftree}
        	\hypo{ (M,R,S,k)\Downarrow(i,R',S') & (M_j,R',S',k)\Downarrow(\bigchi,R_j,S_j) }
        	\infer1[]{ (\texttt{if $M$ then $M_1$ else $M_0$},R,S,k)\Downarrow(\bigchi,R_j,S_j)}
            \end{prooftree}}
            \\
              	(\not\Downarrow_\text{rec})&
    	{\begin{prooftree}
   		\hypo{ (N[m/f],R[m\mapsto\lambda x.M\{m/f\},S,k)\Downarrow(\bigchi,R',S')  }
   		\infer1[]
   		{ (\texttt{letrec $f=\lambda x.M$ in $N$},R,S,k)\Downarrow(\bigchi,R',S') }
              \end{prooftree}}
	\end{tabular}
    \begin{tabular}{RLRL}

    \end{tabular}
	\caption{Bounded operational semantics rules. In all cases, $k\neq\nil$, $\bigchi\in\{\fail,\nil\}$, and $j=0$ if $i=0$, and $j=1$ otherwise.}
	\label{fig:semantics}
\end{figure*}

\paragraph{Nominal determinacy}

While the operational semantics is bounded in depth, the reduction tree of a given term can still be infinite because of the non-determinacy involved in evaluating $\lambda$-abstractions (rule $\Downarrow_\lambda$): the rule non-deterministically creates a fresh name $m$ and extends the repository with $m$ mapped to the given $\lambda$-abstraction. This kind of non-determinism, which can be seen as \emph{determinism up to  fresh name creation}, is formalised below.

Let us consider permutations $\pi:\meths\to\meths$ such that, for all $m$, if $m\in\meths_{\theta\to\theta'}$ then
$\pi(m)\in\meths_{\theta\to\theta'}$.
We call such a permutation $\pi$ \emph{finite} if the set $\{a\mid\pi(a)\neq a\}$ is finite. 
Given a syntactic object $X$ (e.g.\ a term, repository, or store) and a finite permutation 
$\pi$, we write $\pi\cdot X$ for the object we obtain from $X$ if we swap each name $a$ appearing in it with $\pi(a)$. Put otherwise, the operation $\cdot$ is an action from finite permutations of $\meths$ to the set of objects $X$. Given a set $\Delta\subseteq\meths$
and objects $X,X'$, we write $X\sim_\Delta X'$ whenever there exists a finite permutation $\pi$ such that:
\[
\pi\cdot X=X' \land \forall a\in\Delta.\ \pi(a)=a
  \]
and say that  $X$ and $X'$ or \emph{nominally equivalent} up to $\Delta$.

\begin{lemma}\label{lem:NomDet}
Given $(M,R,S,k)\Downarrow(\bigchi,R',S')$, for all 
$(\bigchi',R'',S'')$ we have \\
$(M,R,S,k)\Downarrow(\bigchi',R'',S'')$ 
iff
$(\bigchi,R',S')\sim_{\dom(\Delta)}(\bigchi',R'',S'')$.
\end{lemma}

To illustrate bounding of method application and the use of names in place of methods, we provide the following example.
\begin{example}\label{ex:rec}
	Consider the following recursive higher-order program of \textsc{HORef} (with some syntactic sugar) which we shall unwind with a bound $k=2$.
\begin{lstlisting}[language=lambda]
  r := 0;
  letrec f = lambda x. if x then (r$++$;  f (x - 1))
                 else (lambda y. assert (y = !r + x))
  in	
  let g = f 5 in g 5
\end{lstlisting}
Let us write the above program as $r:=0;M_1$, and $M_1$ as $\texttt{letrec $f=\lambda x.M_2$ in $M_3$}$.
We can attempt to evaluate
$(r:=0;M_1,\varnothing,\varnothing,2)$ as follows. Below we let  $S_i=\{r\mapsto i\}$, and
 $N_i= \lambda y.\,\texttt{assert}(y = {!r} + i)$.
\begin{enumerate}
\item First, $(r:=0,\varnothing,\varnothing,2)\Downarrow((),\varnothing,S_0,2)$,
\item next, evaluate
  $(\texttt{letrec $f=\lambda x.M_2$ in $M_3$},\varnothing,S_0,2)$,
\item i.e.\ $(M_3\{m_1/f\},R,S_0,2)$, with $R=\{m_1 \mapsto \lambda x.M_2\{m_1/f\}\}$,
\item i.e.\ $(\texttt{let $g=m_1\,5$ in $g\, 5$},R,S_0,2)$,
\item now, first evaluate $(m_1\,5,R,S_0,2)$,
\item i.e.\ $(\texttt{if $5$ then $(r{+}{+};m_1(5-1))$ else $N_5$},R,S_0,1)$,
\item i.e.\ $(r{+}{+};m_1(5-1),R,S_0,1)$,
\item i.e.\ $(m_14,R,S_1,1)$,
\item i.e.\ $(\texttt{if $4$ then $(r{+}{+};m_1(4-1))$ else $N_4$},R,S_1,0)$,
\item i.e.\ $(r{+}{+};m_1(4-1),R,S_1,0)$,
\item i.e.\ $(m_13,R,S_2,0)$, 
\item i.e.\ $(\texttt{if $3$ then $(r{+}{+};m_1(3-1))$ else $N_3$},R,S_2,\nil)$,
  \\ which returns $\nil$.
\end{enumerate}%
\cutout{
\begin{xtabular}{R L}
&(r:=0;M_1,\varnothing,\varnothing,2)\\
\hline
&(r:=0,\varnothing,\varnothing,2)\Downarrow((),\varnothing,\{r \mapsto 1\},2)\\
\hline
M_1:&(\texttt{letrec $f=\lambda x.M_2$ in $M_3$},\varnothing,\{r \mapsto 1\},2)\\
&\Downarrow(M_3\{m_1/f\},\{m_1 \mapsto M_2\},\{r \mapsto 1\},2)\\
M3: &\Downarrow(\texttt{let $g=m_1\,5$ in $M_4$},\{m_1 \mapsto M_2\{m_1/f\},...)\\
\hline
&(m_1\,5,\{m_1 \mapsto M_2\},\{r \mapsto 1\},2)\\
&\Downarrow(M_2\{5/x\},\{m_1 \mapsto M_2\},\{r \mapsto 1\},1)\\
M_2: &\Downarrow(\texttt{if $5$ then $(r++;m_1(5-1))$ else $M_5$},...,1)\\
\hline
&(r++;m_1\,4,...,\{r \mapsto 1\},1)\\
&\Downarrow(m_1\,4,...,\{r \mapsto 2\},1)\\
&\Downarrow(M_2\{4/x\},...,\{r \mapsto 2\},0)\\
M_2: &\Downarrow(\texttt{if $4$ then $(r++;m_1(4-1))$ else $M_6$},...,0)\\
\hline
&(r++;m_1\,3,...,\{r \mapsto 2\},0)\\
&\Downarrow(m_1\,3,...,\{r \mapsto 3\},0)\\
&\Downarrow(M_2\{3/x\},...,\{r \mapsto 3\},\nil)\\
&\Downarrow(\nil,...,\{r \mapsto 3\},\nil)\\
\hline
M_6:&(\nil,...,\{r \mapsto 3\},\nil)\\
\hline
M_5:&(\nil,...,\{r \mapsto 3\},\nil)\\
\hline
M_4:&(\nil,...,\{r \mapsto 3\},\nil)\\
\end{xtabular}}%
Hence, the evaluation aborts with $\nil$.
The interesting part of the program is the assertion, which is not reached with this bound.
Setting the bound to 6, $m_1$ will eventually be called with 0 and return the function $N_0$. The latter will be bound to $g$ and called on 5, and at that point $r$ will have value 5, so the assertion will pass. Setting initially $r:=1$ would lead to failure.

Intuitively, the bounded semantics is equivalent to a bounded inlining of methods. As such, evaluating the example with $k=2$ can be seen as unwinding the program as follows (where we have also included the function definitions for clarity).
\begin{lstlisting}[language=LAMBDA]
  r := 0;
  letrec f =
    lambda x. if x then (r$++$;  f (x - 1))
         else (lambda y. assert(y = !r + x))
  in
  let g = if 5 then ( r$++$;
                      if 4 then (r$++$; nil)
                      else (lambda y. assert(y = !r + 4)) )
          else (lambda y. assert(y = !r + 5))
  in if 5 then ( r$++$;
                 if 4 then (r$++$; nil) 
                 else assert(5 = !r + 4) )
     else assert(5 = !r + 5)
\end{lstlisting}\vspace{-2mm}
We will come back to this example in the next section.
\end{example}


\section{A Bounded Translation for \textsc{HORef}}\label{sec:trans}

We present an algorithm which, given a term $M$ and a bound $k$,
produces a propositional formula which captures the bounded semantics of $M$, for the bound $k$. More precisely, 
the algorithm receives a valid configuration $(M,R,S,k)$ as input, where $M$ may only contain free variables of ground type, and produces a formula $\phi$ and a variable $\ret$. Then, for any substitution $\sigma_0$ closing the configuration, and any corresponding formula $\ofSig{\sigma_0}$, $(M,R,S,k)\{\sigma_0\}$ reaches some $\chi\in\vals\cup\{\fail,\nil\}$ iff
\[
(\phi\land\ofSig{\sigma_0})\implies (\ret=\chi)
\]
is satisfiable. The formal statement and proof of the above is given in Theorem~\ref{thm:sound}. It is slightly more elaborate as it takes into account the possibly different choices of fresh method names in the translation and evaluation.

The translation operates on intermediate symbolic configurations, of the form
$(M,R,C,D,\phi,k)$, where:
\begin{itemize}
\item $C,D : \refs \rightharpoonup \ssarefs$
  are static single assignment (SSA) maps where \ssarefs~is the set of SSA variables of the form $r_i$ such that $i$ is the number of times $r$ has been assigned to so far.
The map $C$ is counting all the assignments that have taken place so far in the translation, whereas $D$ only counts those in the current path. E.g.\
$C(r)=r_5$ if $r$ has been assigned to five times so far.
We write $C[r]$ to mean \emph{update $C$ with reference $r$}: if $C(r)=r_i$, then $C[r] = C[r \mapsto r_{i+1}]$, where $r_{i+1}$ is fresh.
	\item $\phi$ is a propositional formula containing the behaviour of the current path so far.
\end{itemize}
Moreover, $R$ is a repository storing all methods created so far, and $k$ is the bound.
The translation returns tuples of the form $(\ret,\phi,R,C,D)$,
where $\phi,R,C,D$ have the same interpretation, albeit for after reaching the end of all  paths for the term $M$.
The variable $\ret$ represents the return value of the initial configuration.

The algorithm uses a fresh-name generator, which is left unspecified but  deterministically produces the next fresh name, or variable, or SSA variable of appropriate type.
Following the SSA approach, the variables $\ret$ in particular are always chosen fresh, so that each $\ret$ identifies a unique evaluation point in the translation.
We use SSA form because it allows us reason about assignment as equations. We compute the SSA form on the fly by substituting all free variables with their corresponding $\ret$ at binding, and through the use of SSA-maps $C$ and $D$ for references. 

\cutout{
The goal of this algorithm is to translate some term $M$ into a formula $\phi$ which captures the behaviour of $M$ up to $k$. This means that the trace produced by evaluating $M$ with $k$-bounded method applications should be included in $\phi$. As with other BMC translations, symbolic approaches are more efficient than evaluating the term, and is why we can call $\phi$ a symbolic trace. Additionally, fresh names are created in some canonical way, by which we mean that the translation follows some deterministic procedure to generate fresh names. Here, the benefits of $\nil$ becomes clearer: it avoids spurious errors and allows full verification by checking if the bound is reachable.}

We now describe the translation. The translation stops when either the bound $\nil$, a $\fail$, or a value $v$ has been reached. The base cases add clauses mapping return variables to actual values of evaluating $M$.

Inductive cases build the symbolic trace of $M$ by recording in $\phi$ all changes to the store, and the return values ($\ret$) at each step in the computation tree. These steps are then chained together using the guard $F$:\vspace{-.5mm}
\begin{align*}
F~a~b~\phi =~ 
&((a=\fail)\Rightarrow(b=\fail)) \land ((a=\nil)\Rightarrow(b=\nil))\\
&{}\land((a=\fail)\lor(a=\nil)\lor \phi)\vspace{-.5mm}
\end{align*}
which propagates $\nil$ and $\fail$, and the SSA maps $C,D$.

The difference between reading ($D$) and writing ($C$) is noticeable when branching.
There are two branching cases here: the conditional case, and the one for application $xM$.
In the former one, we branch according to the return value of the condition (denoted by $ret_b$), and each branch translates $M_0$ and $M_1$ respectively. 
In this case, both branches read from the same map $D_b$, but may contain different assignments, which we accumulate in $C$. The formula $\psi_0\land\psi_1$ encodes a binary decision tree with guarded clauses that represent the path guards.

When applying variables as methods ($xM$, with $x:\theta$), we encode in $\psi$ an $n$-ary decision tree where $n$ is the number of methods to consider. This is necessary since
the algorithm is symbolic and therefore agnostic to what $x$ is pointing at. In such cases, we assume non-determinism, meaning that $x$ could be any method in the repository $R$ restricted to type $\theta$ (denoted $R \upharpoonright \theta$). We call this case \boldemph{non-deterministic method application}. This case seems to be fundamental for applying BMC to higher-order terms, and higher-order references. It is made possible by the introduction of names for methods, as it allows for comparison of higher-order terms as values.
Non-deterministic method application is a primary source of scalability problems, however, and will be discussed in more detail later.

The BMC translation is given as follows. It transforms each symbolic configuration $(M,R,C,D,k)$ to $\sem{M,R,C,D,k}$.
In all of the cases below, $\ret$ is a fresh variable and $k\not=\nil$. We also assume a common domain $\varPi = dom(C) = dom(D)$, which is the finite subset of {\refs} containing all references that appear in $M$ and $R$.
\smallskip

\noindent
\textbf{Base Cases:}
{\setlength{\leftmargini}{0pt}
\begin{itemize}
	\item \tnil
	\item \tfail
	\item \tval
	\item \tderef
	\item \tlambda
\end{itemize}

\noindent
\textbf{Inductive Cases:}
\begin{itemize}
	\item \tpi
	\item \tplus
	\item \tpair
	\item \tletin
    \item \tletrec
	\item \tapplym
	\item \tifthenelse
	\item \tapplyx
	\item \tassign
\end{itemize}}

To illustrate the algorithm, we look at two characteristic cases. In $\sem{\letin{x=M}\,M',R,C,D,\phi,k}$, we first compute the translation of $M$. Using the results of said translation, we can substitute in $M'$ the fresh variable $\ret_1$ for $x$, and compute its translation. To finish, we return $\ret$, chain it to $\ret_2$ using predicate $F$ in $\phi_2$, and return the remaining results of translating $M'$.

In
$\sem{xM,R,C,D,\phi,k}$ we see 
non-deterministic method application in action. We first translate the argument $M$ and obtain $(\ret_0,\phi_0,R_0,C_0,D_0)$. We then restrict the repository $R$ to type $\theta$ to obtain the set of names identifying all methods of matching type for $x$. If no such methods exist, this means that the binding of $x$ had not succeeded (because of $\fail/\nil$) and we are examining a dead branch, so we immediately return. Otherwise, for each method $m_i$ in this set, we obtain the translation of applying $m_i$ to the argument $\ret_0$. This is done by substituting $\ret_0$ for $y_i$ in the body of $m_i$. After translating all method applications, all paths are joined in $\psi$, as described earlier, by constructing an $n$-ary decision tree that includes the state of the store in each path. We do this by incrementing all references in $C_n$, and adding the clauses $C_n' = D_i(r)$ for each path. These paths are then guarded by the clauses $(x = m_i)$. Finally, we return a formula that propagates $\nil$ and $\fail$ in case $\ret_0$ reaches either of them. Note that we return $C_n'$ as both the $C$ and $D$ resulting from translating this term. This is because all branches have been joined, and any term sequenced after this one should have all updates available to it.

\cutout{
A proof of soundness of the translation is provided in the next section. This theorem states that, given a possibly open configuration, the translation must agree with the operational semantics under permutation of names--meaning that, by determinacy of the operational semantics, the results should be equivalent given the choice of fresh names--and vice versa. This stems from two properties of the translation: correctness and uniqueness. Uniqueness states that any given closed configuration should produce a formula satisfiable only with a unique assignment for the result, while correctness states that said assignment must be equivalent to the evaluation of the corresponding operational semantics.
}

We now come back to Example~\ref{ex:rec} to illustrate the intuition of SSA and return variables, non-deterministic method application, and formula construction.
\begin{example} Consider Example~\ref{ex:rec} modified with free variables $n$ and $r_0$.
\begin{lstlisting}[language=lambda]
  r := r0;
  letrec f = lambda x. if x then (r$++$;  f (x - 1))
                  else (lambda y. assert (y = !r + x))
  in	
  let g = f n in g n
\end{lstlisting}
 We transform it to produce the program in SSA form with non-deterministic method application at line 11, again unwinding with $k=2$. Note that all assignments have been replaced with let-bindings. This is because, in SSA form, we think of references as SSA variables. In addition, we use keyword \texttt{new} to add new for names to the repository.
\begin{lstlisting}[language=LAMBDA,numbers=left,xleftmargin=1em]
let r1 = r0 in
letrec m1 =
  lambda x. if x then (r$++$;  m1 (x-1))
       else (lambda y. assert(y = !r + x)) 
in
let ret3 =
  if n then (let r2 = r1 + 1 in
    if n-1 then (let r3 = r2 + 1 in nil)
    else (new m3 = lambda y. assert(y = !r+n-1) in m3))
  else (new m2 = lambda y. assert(y = !r + n) in m2)
in match ret3 with
| nil rarr nil
| m3  rarr assert(n = r3 + n-1)
| m2  rarr assert(n = r3 + n)
\end{lstlisting}
We can then build model $\phi$ for the example. For economy, we hide the nuances of propagating $\nil$ and $\fail$ in predicate $F_{a,b}P$, which is short-hand for $F~\ret_a~\ret_b~P$. We also omit $F$ wherever no $\fail$ or $\nil$ appears in the term, and directly return constants instead of translating them. To construct the formula, we traverse the term in order, and add clauses in order of traversal. Note that the ``else" branch is always explored first in conditionals.
\begin{flalign*}
\phi = & F_{2,1}(\ret_1 = \ret_2)\land(r_1=r_0)&&\text{(line 1)}\\
       & {}\land F_{9,2}(\ret_2 = \ret_9)&&\text{(line 6)}\\
       & {}\land F_{4,3}(F_{5,3}(\ret = (n=0)?\ret_4:\ret_5))&&\text{(line 7)}\\
       & {}\land (\ret_4 = m_2)&&\text{(line 10)}\\
       & {}\land F_{6,5}(\ret_5 = \ret_6)\land(r_2=r_1+1)&&\text{(line 7)}\\
       & {}\land F_{7,6}(F_{8,6}(\ret_6 = ((n-1)=0)?\ret_7:\ret_8))&&\text{(line 8)}\\
       & {}\land (\ret_7 = m_3)&&\text{(line 9)}\\
       & {}\land F_{9,8}(\ret_8 = \ret_9)\land(r_3=r_2+1)&&\text{(line 8)}\\
       & {}\land F(\ret_9 = \nil)&&\text{(line 8)}\\
       & {}\land ((\ret = m_3) \Rightarrow (\ret_{10} = \ret_{11}))&&\text{(line 13)}\\
       & {}\land ((\ret_{11} = ((n=r_3+n-1)=0)?\fail:())&&\text{(line 13)}\\
       & {}\land ((\ret = m_2) \Rightarrow (\ret_{10} = \ret_{13}))&&\text{(line 14)}\\
       & {}\land (\ret_{13} = ((n=r_3+n)=0)?\fail:())&&\text{(line 14)}
\end{flalign*}

In this case, if we set $r_0 = 0$, recalling $k=2$, then $\phi \land (\ret_1=\nil)$ is satisfiable with a minimum $n=2$, since we need at least $k=n+1$ iterations to reach $n=0$ (which is also the case for a negative $n$, as the program
would diverge). With $r_0=0$, however, we cannot violate the assertion, i.e. $\phi \land (\ret_1 = \fail)$ is not satisfiable. Setting $r_0 = 1$, on the other hand, causes $\phi \land (\ret_1=\fail)$ to be satisfiable with $n=0$ or $n=1$.
\end{example}

\paragraph{Bounded Model Checking with the Translation}
The steps to do a $k$-bounded model checking of some configuration $(M,R,S,k)$ using the bounded translation algorithm described previously are as follows:
\begin{enumerate}
	\item Build the initial axioms/preconditions:
    \\$\phi_0 = \bigwedge_{r \in S}(r = S(r))$.
	\item Build the initial SSA maps: 
    \\$C_0 = \{r \mapsto r_0 \mid r \in dom(S)\}$.
	\item Compute the translation: 
    \\$\llbracket M_0,R,C_0,C_0,\phi_0,k \rrbracket = (\ret,\phi,R',C,D)$.
	\item This is where the expressiveness of $\fail$ and $\nil$ becomes relevant. To check for:
	\begin{enumerate}
		\item sound errors: 
        \\$\phi' = (\ret = \fail) \land \phi$
		\item reached bounds (for verification): 
        \\$\phi' = (\ret = \nil) \land \phi$
		\item a specific return property $p$: 
        \\$\phi' = \neg p \land \phi$, e.g. $\phi' = \neg (\ret > 5) \land \phi$
		\item a specific store property $p_r$: 
        \\$\phi' = \neg p_r \land \phi$, e.g. $\phi' = \neg (D(r) > 5) \land \phi$
	\end{enumerate}
	\item Transform $\phi'$ to the relevant SMT solver format (e.g. SMT-Lib), and use the SMT solver to get a satisfying assignment.
    \item When checking for $\fail$, if the formula is unsatisfiable, we can increase the bound given checking for $\nil$ is satisfiable. If $\nil$ is not satisfiable either, then the program has been verified.
\end{enumerate}
Note that checks for store properties (d) can be combined with any of the properties mentioned in step (6), including other store properties. It is only possible to check other properties (a,b, and c) independent of each other, however. This is because the translation is deterministic and will always output a unique result. For instance, the return value cannot be both $\fail$ and $\nil$ in the same satisfying assignment, i.e. the formula $\phi \land (\ret = \nil) \land (\ret = \fail)$ is unsatisfiable. Moreover, while the semantics requires closed terms, the translation is indifferent towards free variables. As such, it will handle top-level input arguments and said free variables by simply adding them into the formula, which will produce a unique return for each valid assignment of the input arguments. This is, in fact, one of the most useful applications of BMC, since it then generates counter-examples from said input arguments. These free variables, however, must be of ground type. The translation will not mind if a free variable is given a higher-order type. But then the resulting formula becomes unsound, since we do not model unknown program code. The simplest solution is to make the formula always unsatisfiable to avoid spurious errors. Handling open terms with higher-order free variables will be discussed in more detail as future work.


\section{Soundness of the BMC translation}

In this section we prove that our BMC algorithm is sound for input terms that are closed or contain open variables of ground type.

We start off with some definitions.
%
An \boldemph{assignment} $\sigma:\vars\rightharpoonup\cvals$
is a finite map from variables to closed values. 
%
Given a term $M$, we write $M\{\sigma\}$ for the term obtained by applying $\sigma$ to $M$. On the other hand, applying $\sigma$ to a method repository $R$, we obtain the repository $R\{\sigma\}=\{m\mapsto R(m)\{\sigma\}\mid m\in\dom(R)\}$ -- and similarly for stores $S$.
Then, given a valid configuration $(M,R,S,k)$, we have
$(M,R,S,k)\{\sigma\}=(M\{\sigma\},R\{\sigma\},S\{\sigma\},k)$.

Given a formula $\psi$ and an assignment $\sigma$, we say $\sigma$ \boldemph{represents}  $\psi$, and write  $\sigma \preserves \psi$, if:
\begin{itemize}
		\item $\sigma$ satisfies $\psi$ (written $\sigma \vDash \psi$);
		\item $\psi$ implies $\sigma$: $\forall x \in dom(\sigma).\ \psi \implies x = \sigma(x)$.
	\end{itemize}
        Given assignment $\sigma$, we define a formula $\ofSig{\sigma}$ representing it by:
        $\ofSig{\sigma} = \bigwedge_{x \in \dom(\sigma)}(x = \sigma(x))$.


Given a valid configuration  $(M,R,S,k)$, let us set:
\begin{align*}
C_S &= \{ r \mapsto r_0 \mid r \in dom(S)\} &
  \phi_S &= \bigwedge\nolimits_{r \in dom(S)}(r_0 = S(r))
\end{align*}
and define: \
$\sem{ M,R,S,k } = 
\sem{ M,R,C_S,C_S,\phi_S,k }$.

\begin{theorem}[Soundness]\label{theorem:sound:open}\label{thm:sound}
Given a valid configuration  $(M,R,S,k)$ whose open variables are of ground type, suppose $\sem{ M,R,S,k } = (\ret,\phi,R'',C,D)$.
Then, for all assignments $\sigma_0$ closing $(M,R,S,k)$, the following are equivalent:
\begin{enumerate}
\item $(M,R,S,k)\{\sigma_0\} \Downarrow (\bigchi', R', S', k')$
\item $\exists\bigchi \sim_{dom(R)} \bigchi'.\ (\phi \land \ofSig{\sigma_0}) \implies (\ret = \bigchi)$.
  \end{enumerate}
  \end{theorem}
\begin{proof}
Take $\sigma_S=\{r_0\mapsto S(r)\{\sigma_0\}\mid r\in\dom(S)\}$ and $\sigma_0'=\sigma_0\cup\sigma_S$. By construction then, $\phi_S\land\ofSig{\sigma_0}\cong\sigma_0'$. 
Let us set $\phi'=\phi\land\ofSig{\sigma_0}$.
From the assumption and the fact that the translation propagates the formula $\ofSig{\sigma_0}$ from the initial condition (Lemma~\ref{lem:4}), we have:
\[
\sem{ M,R,C_S,C_S,\phi_S\land\ofSig{\sigma_0},k }
=(\ret,\phi',R'',C,D)
  \]
Moreover,  $(M,R,C_S,k)\{\sigma_0'\} = (M,R,S,k)\{\sigma_0\}$ is valid and closed so, by Lemma~\ref{lem:1}, we have
that  $(M,R,S,k)\{\sigma_0\} \Downarrow (\hat \bigchi, \hat R,\hat S,\hat k)$\
and $\exists \sigma_1 \supseteq \sigma_0'.\ (\sigma_1 \preserves \phi')\land\phi' \implies (\ret = \hat{\bigchi})$.

Suppose now~$(1)$ holds. By Lemma~\ref{lem:NomDet}, we have that
  $\bigchi'\sim_{dom(R)} \hat\bigchi$, so  taking $\bigchi=\hat\bigchi$ we obtain~(2).

  On the other hand, if~$(2)$ holds then
  $\phi'$ implies $\ret=\bigchi$ and $\ret=\hat\bigchi$.
  Since $\phi'\cong\sigma_1$, we get $\bigchi= \sigma_1(\ret) =\hat\bigchi$. Hence,
  $\bigchi'\sim_{\dom(R)}\hat\bigchi$ and we conclude using Lemma~\ref{lem:NomDet}.
	\end{proof}

Theorem~\ref{thm:sound} uses the following main lemma, which is shown in the Appendix.
Below we assume that
\[
\sigma,\sigma':(\vars\cup\ssarefs)\rightharpoonup\cvals\cup\{\fail,\nil\}
\]
 are \emph{extended assignments}. Accordingly,
an \emph{extended term}
is a term that may contain $\nil$, $\nil\,M$ or $\fail\,M$ as a subterm (extended terms are closed under extended assignments). Extended configurations are defined in a similar manner, and we use the same operational semantics rule to evaluate them. In particular, extended configurations may contain extended terms with free variables, or $\fail/\nil$, in evaluating position. We do not add special rules for those\,--\,they get stuck.

\begin{lemma}[Correctness]\label{lem:1}
  Given $M,R,C,D,k,\phi,\sigma$ such that $\sigma\preserves\phi$ and\\
 $(M,R,D,k)\{\sigma\}$ is terminating,
 if $\llbracket M,R,C,D,\phi,k \rrbracket =(\ret,\phi',R',C',D')$
then:
\begin{itemize}
	\item $(M,R,D,k)\{\sigma\} \Downarrow (\bigchi,\hat{R},\hat S,\hat{k})$ and $\exists \sigma' \supseteq \sigma. (\sigma' \preserves \phi')$ and $\phi' \implies (\ret = \bigchi)$, and
	\item if \emph{$\bigchi \notin \{\fail,\nil\}$}, then
          $\hat R \subseteq R'\{\sigma'\}$ and $\hat S = D'\{\sigma'\}$.
\end{itemize}
\end{lemma}


\section{A Points-to Analysis for Names}\label{sec:points-to}

The presence of non-deterministic method application in our BMC translation is a primary source of combinatorial explosion of the algorithm. As such, a more precise filtering of $R$ is necessary for scalability. 
In this section we describe an optimisation on non-deterministic method application inspired by points-to analysis. Points-to analysis
provides an overapproximation
of the \emph{points-to set} of each variable inside a program, that is, the set of locations that it may point to. 
Here instead we devise an analysis that overapproximates the set of methods that may be bound to each variable in the bounded unfolding of a program in SSA form. This way, we can reduce the branching caused by non-deterministic method application in the BMC translation.

Traditionally, the problem of which method to apply per method application is one that CFA~\cite{DBLP:books/daglib/0098888,DBLP:conf/pldi/Shivers88,DBLP:journals/csur/Midtgaard12} answers. In our setting, however, methods are represented by names, which reduces the task of higher-order method application to keeping track of names used in a first-order term (and, additionally, names can be stored).
This allows us to address the same problem as CFA with a simpler points-to analysis for names. Note that this program analysis is not as expressive as full points-to analysis, in the sense that we do not need to consider pointers pointing at pointers.

Points-to analysis
algorithms often belong to one of two families: \textit{Steengaard-style}~\cite{DBLP:conf/popl/Steensgaard96} and \textit{Andersen-style}~\cite{Andersen:94:PhD}, also known as \textit{unification-based} and \textit{inclusion-based} flow-insensitive analyses respectively~\cite{DBLP:conf/paste/Hind01}.
These are typically constrain-based analyses whereby one goes through the code of a program and allocates constraints to each reference/variable assignment, and subsequently solves them in a global manner.
In our case,
the BMC translation already performs a recursive analysis on terms, which we can use to make the points-to analysis more precise and local, while remaining efficient.

\cutout{
For example, let $pts(r)$ be the points-to set of $r$. The following are some cases and corresponding constraints for ANSI-C programs in Andersen-style analysis:
\begin{align*}
\texttt{p := \&x} \qquad & x \in pts(p)\\
\texttt{p := q}   \qquad & pts(p) \supseteq pts(q)\\
\texttt{*p := q}  \qquad & \forall r \in pts(p). pts(r) \supseteq pts(q)\\
\texttt{p := *q}  \qquad & \forall r \in pts(q). pts(p) \supseteq pts(r)
\end{align*}
We draw inspiration from points-to analysis to reduce the set of names to consider in non-deterministic method application. The idea is to compute points-to sets that tell us which names can be pointed at by a given variable or reference.
}

The analysis looks at references and variables, and assigns to them a set of method names that they may be referring/bound to.
This is done via a finite map
\[
  pt:(\refs\cup\vars) \rightharpoonup \pts
\]
where $\pts$ contains all \emph{points-to sets} and is given by:
\[
\pts\ \ni \ A\, ::=\ X \mid \pair{A,A} \quad (\text{where } X\subseteq_{fin}\meths).
\]
Thus, a points-to set is either a finite set of names or a pair of points-to sets.
In the BMC translation, 
points-to sets need to be created when a method name is created. Moreover, they need to be assigned to references or variables in the following cases:
\begin{align*}
&r:=M \quad && \text{add in $pt$: } r\mapsto pt(M)\\
&\letin{x=M}\,M'   \quad && \text{add in $pt$: }  x\mapsto pt(M)\\
&xM  \qquad && \text{add in $pt$: } \ret(M)\mapsto pt(M)
\end{align*}
where $\ret(M)$ is the variable assigned to the result of $M$.
The $\texttt{letrec}$ follows a similar logic.
The need to have sets of names, instead of single names, in the range of $pt$ is that the analysis, being symbolic, branches on conditionals and applications, so the method pointed to by a reference cannot be decided during the analysis. Thus, when joining after branching, we merge the $pt$ maps obtained from all branches. 

\cutout{
  In \textsc{HORef}, references are the only mutable construct, and for this optimisation we only care about names. This makes our ``points-to analysis" simpler than those for ANSI-C. In the translation we deal with higher-order terms using names, and use return variables $\ret$ to identify steps in computation. Since return variables are semantically indistinguishable from normal variables that appear in the term, all bound variables are substituted for their corresponding $\ret$. We thus need relations mapping $\ret$ variables and references to their points-to sets. }

The points-to algorithm is presented next. Given a valid configuration $(M,R,S,k)$, the algorithm returns $PT(M,R,S,k)=(\ret,A,R,pt)$, where $A$ is the points-to set of $\ret$, and $pt$ is the overall points-to map computed. 
The 
 union operator for two points-to sets of matching form is:
\begin{align*}
&A \cup B=
\begin{cases}
\pair{A_1 \cup B_1, A_2 \cup B_2} &\text{if }A,B = \pair{A_1,A_2},\pair{B_1,B_2}\\
A \cup B &\text{if }A,B \subseteq\meths
\end{cases}
\end{align*}
while the merge of points-to maps is given by:
$$
merge(pt_1,\dots,pt_n) = \{x \mapsto \bigcup\nolimits_i\hat{pt}_i \mid x\in\bigcup\nolimits_i\dom(pt_i)\}
$$
where $\hat{pt}_i(x)=pt_i(x)$ if $x\in dom(pt_i)$, and $\emptyset$ otherwise.
\smallskip

\noindent
\textbf{Base Cases:}
{\setlength{\leftmargini}{0pt}
  \begin{itemize}
    \item $PT(M,R,pt,nil) = (ret,\varnothing,R,pt)$
    \item $PT(fail,R,pt,k) = (ret,\varnothing,R,pt)$
    \item $PT(v,R,pt,k) = (ret,\varnothing,R,pt)$, where $v=i,()$
    \item $PT(m,R,pt,k) = (ret,\{m\},R,pt)$
    \item $PT(x,R,pt,k) = (ret,pt(x),R,pt)$
    \item $PT({!r},R,pt,k) = (ret,pt(r),R,pt)$
    \item $PT(\lambda x.M,R,pt,k) = (ret,\{m\},R[m \mapsto \lambda.M],pt)$
\end{itemize}
\smallskip

\noindent
\textbf{Inductive Cases:}
\begin{itemize}
	\item 
	$\begin{aligned}[t]
		& PT(\pi_i\,M,R,pt,k) =\\
		&\quad \begin{aligned}[t]
			& \letin{(\ret_1,A_1,R_1,pt_1) = PT(M,R,pt,k)}\\
			& (\ret,\pi_i\,A_1,R_1,pt_1)
		\end{aligned}
	\end{aligned}$
	
	\item 
	$\begin{aligned}[t]
		& PT(r:=M,R,pt,k) =\\
		&\quad \begin{aligned}[t]
			& \letin{(\ret_1,A_1,R_1,pt_1) = PT(M,R,pt,k)}\\
			& (\ret,\varnothing,R_1,pt_1[r \mapsto A_1])
		\end{aligned}
	\end{aligned}$
	
	\item 
	$\begin{aligned}[t]
		& PT(M_1 \oplus M_2,R,pt,k) =\\
		&\quad \begin{aligned}[t]
			& \letin{(\ret_1,A_1,R_1,pt_1) = PT( M_1,R,pt,k)}\\
			& \letin{(\ret_2,A_2,R_2,pt_2) = PT( M_2,R_1,pt_1,k)}\\
			& (\ret,\varnothing,R_2,pt_2)
		\end{aligned}
	\end{aligned}$

	\item 
	$\begin{aligned}[t]
		& PT(\pair{M_1,M_2},R,pt,k) =\\
		&\quad \begin{aligned}[t]
			& \letin{(\ret_1,A_1,R_1,pt_1) = PT( M_1,R,pt,k)}\\
			& \letin{(\ret_2,A_2,R_2,pt_2) = PT( M_2,R_1,pt_1,k)}\\
			& (\ret,\pair{A_1,A_2},R_2,pt_2)
		\end{aligned}
	\end{aligned}$

	\item 
	$\begin{aligned}[t]
		& PT(\texttt{let $x=M$ in $M'$},R,pt,k) =\\
		&\quad \begin{aligned}[t]
			& \letin{(\ret_1,A_1,R_1,pt_1) = PT(M,R,pt,k)}\\
			& PT(M'\{\ret_1/x\},R_1,pt_1[\ret_1 \mapsto A_1],k)
		\end{aligned}
	\end{aligned}$

	\item 
	$\begin{aligned}[t]
		& PT( \texttt{letrec $f=\lambda x.M$ in $M'$},R,pt,k) =\\
		&\quad 
        \begin{aligned}[t]
			&\letin{\text{$m,f'$ be fresh}}\\
            & PT( M'\{f'/f\},R[m \mapsto \lambda x.M\{f'/f\}],
             pt[f' \mapsto \{m\}],k)
		\end{aligned}
	\end{aligned}$

	\item 
	$\begin{aligned}[t]
		& PT( m\,M,R,pt,k) =\\
		&\quad \begin{aligned}[t]
			& \letin{(\ret_1,A_1,R_1,pt_1) = PT(M,R,pt,k)}\\
			& \letin{R(m)\text{ be }\lambda x.N}\\
            & PT(N\{\ret_1/x\},R_1,pt_1[\ret_1 \mapsto A_1],k)
		\end{aligned}
	\end{aligned}$	
	
	\item
	$\begin{aligned}[t]
		& PT( \texttt{if $M_b$ then $M_1$ else $M_0$},R,pt,k) =\\
		&\quad \begin{aligned}[t]
			& \letin{(\ret_b,A_b,R_b,pt_b) = PT(M_b,R,pt,k)}\\
			& \letin{(\ret_0,A_0,R_0,pt_0) = PT(M_0,R_b,pt_b,k)}\\
			& \letin{(\ret_1,A_1,R_1,pt_1) = PT(M_1,R_0,pt_b,k)}\\
			& (\ret,A_0 \cup A_1,R_1,merge(pt_0,pt_1))
		\end{aligned}
	\end{aligned}$
	
	\item
	$\begin{aligned}[t]
		& PT( x^\theta\,M,R,pt,k) =\\
		&\quad \begin{aligned}[t]
			& \letin{(\ret_0,A_0,R_0,pt_0) = PT( M,R,pt,k)}\\
			& \letin{pt(x)\text{ be }\{m_1,...,m_n\}}\\
                        & \text{if $n=0$ then $(\ret,\emptyset,R_0,pt_0)$ else:}\\
			& \letin{pt_0' = pt_0[\ret_0 \mapsto A_0]}
			\ \text{for each }i \in \{1,...,n\}:\\
			&\quad 
			\begin{aligned}[t]
				& \letin{R(m_i)\text{ be }\lambda y_i.N}\\
				& \letin{(\ret_i,A_i,R_i,pt_i) = 
                 PT(N_i\{\ret_0/y_i\},R_{i-1},pt_0',k)}\\
			\end{aligned}\\
			& (\ret,A_1 \cup ... \cup A_n,R_n,merge(pt_1, \dots, pt_n))
		\end{aligned}
	\end{aligned}$
      \end{itemize}}
\smallskip

\paragraph{The optimised BMC translation}
We can now incorporate the points-to analysis in the BMC translation to get an optimised translation which operates on symbolic configurations augmented with a points-to map, and returns:
\[
\sem{M,R,C,D,pt,\phi,k}_{PT} = (\ret,\phi',R',C',D',A,pt')
\]
The optimised BMC translation is defined by lock-stepping the two algorithms presented above (i.e.\ $\sem{\_}$ and $PT(\_)$) and let $\sem{\_}$ be informed from $PT(\_)$ in the $xM$ case, which now restricts the choices of names for $x$ to the set $pt(x)$.
We give the full algorithm in Appendix~\ref{apx:full}.
Its soundness is proven along the same lines as the basic algorithm.

To illustrate the significance of reducing the set of names, we provide a simple example.

\begin{example}
	Consider the following program which recursively generates names to compute triangular numbers.
	\begin{lstlisting}[language=Lambda]
  letrec f = lambda x. 
    if x leq 0 then 0
    else let g = (lambda y.x + y) in g (f (x-1)) 
  in
  letrec f' = lambda x.if x leq 0 then 0 else x + (f' (x-1)) 
  in assert(f n = f' n)
	\end{lstlisting}
    Without points-to analysis, since $f$ creates a new method, and the translation considers all methods of matching type per recursive call, the number of names to apply at depth $m\leq n$ when translating $f(n)$ is approximately ${m!}$. This means that the number of paths explored grows by the factorial of $n$, with the total number of methods created being the left factorial sum ${!n}$, and total number of names considered being the derangement of $n$. In contrast, $f'(n)$ only considers $n$ names with a linear growth in number of paths. With points-to analysis, the number of names considered and created in $f$ is reduced to that of $f'$.
\end{example}

\paragraph{Other optimisations}

There are other notable deficiencies in our translation. The first one involves our SSA transformation when branching. Particularly, this focuses on the join operations performed, since these add several clauses analogous to $\phi$ functions in conventional SSA. In our approach, the store is updated after branching, which serves as the join. This joining step is not very efficient as it adds a guarded clause per reference per branch. For this, including insights from standard SSA transformations may improve our translation. For example, one naive way to improve performance (and decrease the size of the model) may be to accumulate all changes, so we know exactly which references to update. A more precise way would be to use dominators: an efficient dominance algorithm would tell us which references may have been updated in order to reach some point in a program. Finally, we can make use of Data-Flow Analysis--in particular, Liveness Analysis--to compute whether references are live or not. This could reduce the number of join operations since we do not need to add clauses for dead references. This can even be further expanded into dead code or dead store elimination, which are useful standard optimisations in general.

The second deficiency is repository redundancy, which occurs when adding new names into the method repository. One can imagine a program that creates multiple copies of the same method, adding a new name to the repository each time. Ideally, if a method is already present in the repository, we should not create a new name. To address this, we can search the repository for a structurally equivalent method when attempting to add a new name to it. This looks for the existence of a name with an $\alpha$-equivalent method body which we can use instead of creating a new unnecessary name. One can even augment this solution by using extensionality instead of structural equality.

Finally, there is the minor problem of unnecessary propagation clauses in $F$. For instance, the translation guards every $\ret$ with a predicate $F$, creating many unnecessary clauses. Given we must traverse the term, we should know whether $\fail$ or $\nil$  are reachable in a term, which allows us to prune many of the propagating clauses in $F$. To do this, we can return a variable $q$ in a four-valued logic for set $Q = \{0,Nil,Fail,Both\}$, where $q \in Q$ is an overapproximation for the reachability of $\fail$ and $\nil$ in a given term. In branching, two variables $q_0,q_1 \in Q$ can be combined by the commutative operation $q_0 + q_1$, which follows the equalities: $q + 0 = q$, $q + q = q$, and $Nil + Fail = Both$, for any $q \in Q$. With this, we add in $F$ the guards corresponding to the $q$ returned. Specifically, if $q=Fail$, we only add the clauses that propagate $\fail$, and similarly for $q=Nil$. For $q=Both$, we add clauses for both $\nil$ and $\fail$, while adding no propagating clauses when $q=0$.


\section{Implementation and Experiments}\label{sec:tool}

We implemented the translation algorithm in a prototype tool to model check higher-order programs called \textsc{BMC-2}\cite{hobmc}. The tool takes program source code written in an \textsc{ML}-like language based on \textsc{OCaml}, and produces a propositional formula in \textsc{SMT-LIB 2} format. This can then be fed to an SMT solver such as \textsc{Z3}. Syntax of the input language is based on the subset of \textsc{OCaml} that corresponds to \textsc{HORef}. Differences between \textsc{OCaml} and our concrete syntax are for ease of parsing and lack of type checking. For instance, all input programs must be either written in ``Barendregt Convention'', meaning all bound variables must be fresh, or such that variables have the same type globally. Additionally, all bound variables are annotated with types, as is left and right projection. Internally, \textsc{BMC-2} implements an abstract syntax that extends \textsc{HORef} with vector arguments. This means that functions can take multiple arguments at once. Intuitively, this is equivalent to adding let-bindings that apply each argument individually. We also implemented the optimisation to avoid unnecessary propagation clauses, as previously described. The tool itself is written in \textsc{OCaml}.

To illustrate our input language, following is a sample program \texttt{mc91-e} from the \textsc{MoCHi} benchmark in \textsc{OCaml}. The keyword \texttt{Methods} is used to define all methods in the repository. The keyword \texttt{Main} is used to define the main method.
{\small\begin{lstlisting}[language=ML]
 Methods:
 mc91 (x:Int) :(Int) =
   if x >= 101 then
     x + -10
   else
     mc91 (mc91 (x + 11));
 
 Main (n:Int) :(Unit):
   if n <= 102 
     then assert ((mc91 n) == 91)
     else skip
\end{lstlisting}}
For this sample program, our tool builds a translation with $k=1$ for which \textsc{Z3} correctly reports that \texttt{fail} is reachable if $n=102$. Details about experiments will be provided later.

We tested our algorithm on 20 sample programs selected from the \textsc{MoCHi} benchmark\cite{DBLP:conf/pepm/SatoUK13}. The programs were translated to our input language and checked using our tool and \textsc{Z3}. Care was taken to keep all sample programs as close to the original source code as our concrete syntax allows. All experiments ran on a machine equipped with an Intel Core i7-6700 CPU clocked at 3.40GHz and 16GB RAM.
All tests were set to time-out at 10 seconds, and up to a maximum bound $k=10$. These limits were chosen due to the combinatorial nature of model checking and the sample programs used. Since an increase in the bound increases the state space exponentially, results beyond 10 seconds did not seem as interesting. \textsc{BMC-2} ran three times per program per bound.
Firstly, we measure performance of the base algorithm in terms of total time spent checking each sample program. Secondly, we measure performance of the translation with added points-to analysis for names and the difference it makes. Finally, in addition to run time, if an error exists in a sample program, the lowest bound needed to find a counterexample is recorded.

\paragraph{Results and Observations}
Figure~\ref{fig:result:base} plots the average time taken for \textsc{BMC-2} to check the \textsc{MoCHi} benchmark programs. Table~\ref{table:pointsto} records the percentage difference in average time taken per bound between \textsc{BMC-2} without points-to analysis and \textsc{BMC-2} with points-to analysis. From these, one can see a dramatic improvement in scalability of programs with non-deterministic method application. For example, two previously infeasible programs \texttt{hrec} and \texttt{hors} timed out at $k=3$ and $k=4$ respectively. With points-to analysis, \texttt{hrec} times out at $k=9$, while \texttt{hors} does not time out. In fact, \texttt{hors} can be checked within the allocated time for bounds upwards of $k=200$.

While extending \textsc{BMC-2} with points-to analysis has some overhead, the overhead appears to be minimal or even negligible for this benchmark. The addition of points-to analysis only negatively affected programs with non-deterministic method application at lower bounds, while leaving programs without said method application unaffected. This effect can be seen on Table~\ref{table:pointsto}, which shows a maximum increase in average execution time of $2.4\%$ for $k=1$. On average, however, execution time decreased by $55.8\%$, with a maximum decrease of $96.8\%$ for $k=3$. Only $k=1$ was negatively affected.

We can also observe that performance of the \textsc{BMC-2} heavily depends on the program it is checking. This makes the possibility of full verification entirely dependent on the nature of the program. For example, \texttt{ack}, which is an implementation of the Ackermann function, is a deeply recursive program, and thus cannot be translated by our algorithm any better than its normal growth. This agrees with the intuition that BMC is not appropriate to find bugs in deep recursion. As mentioned before, however, BMC has been shown empirically effective on shallow bugs in industry. To show this, Table~\ref{table:pointsto} records the minimum bound required for \textsc{BMC-2} to find a counterexample. We can observe that all bugs were shallow; occurring within $k=2$. We can thus say that BMC is a very inexpensive technique to find bugs in this benchmark.

\begin{figure}
	\includegraphics[width=\linewidth]{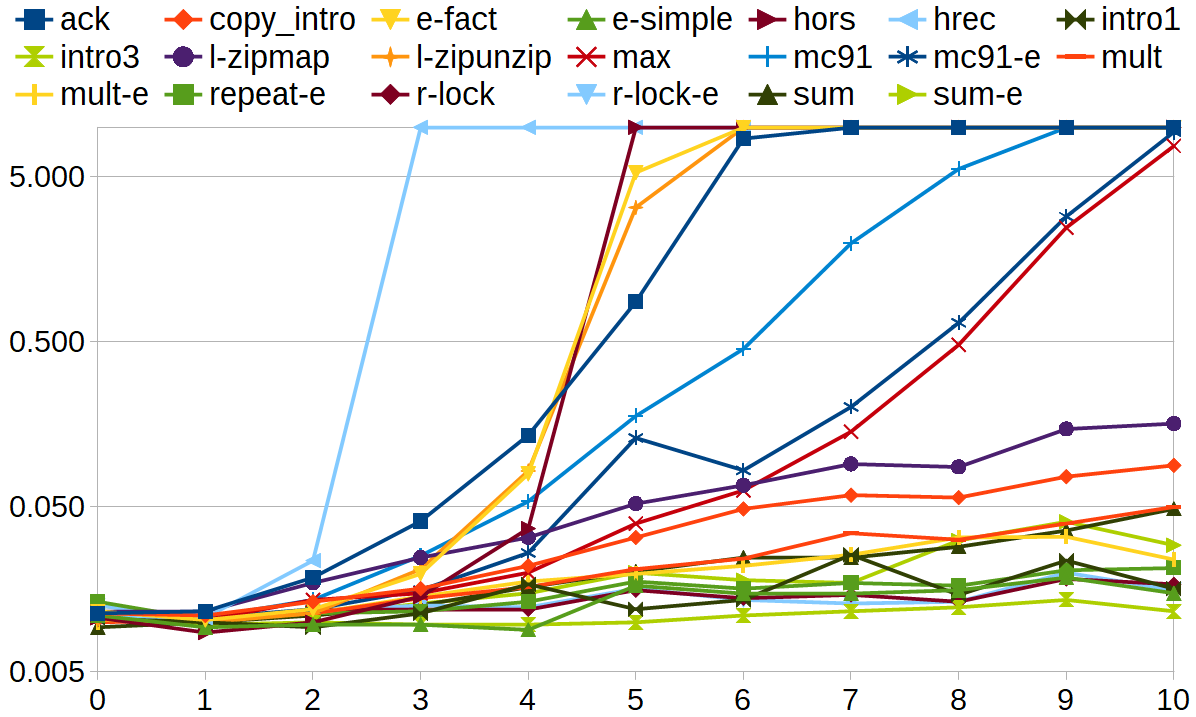}
	\includegraphics[width=\linewidth]{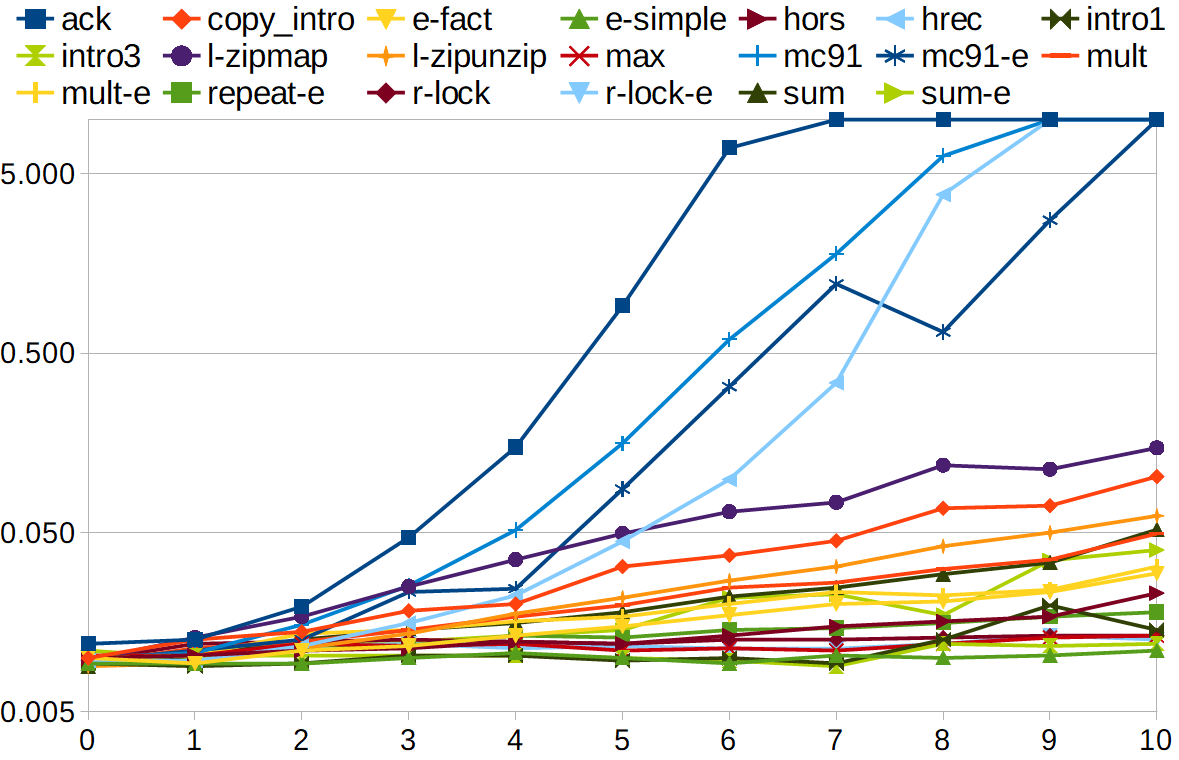}
	\caption{Average total execution times (s) with (bottom) and without (top) points-to analysis on bounds $k=0\dots 10$}
	\label{fig:result:base}
\end{figure}

\begin{table}
	\caption{Smallest bound needed for \textsc{BMC-2} with points-to analysis to find a counterexample (left). Percentage change in average time taken per bound after adding points-to analysis (right).}
	\label{table:pointsto}
	\begin{center}
	\begin{tabular}{ l r r }
		\hline			
		Program & $k$& time \\ \hline
		\texttt{mc91-e} & 1 & 0.011 \\
		\texttt{mult-e} & 1 & 0.010 \\
		\texttt{repeat-e} & 1 & 0.010 \\
		\texttt{r-lock-e} & 2 & 0.012 \\
		\texttt{sum-e} & 1 & 0.010 \\
		\hline			  
	\end{tabular}
    \quad
    	\begin{tabular}{ c r c r }
    		\hline			
    		$k$& \%$\Delta$ & $k$ & \%$\Delta$\\ \hline 
    		0& -10.928     & 6   & -83.172  \\
    		1& 2.412       & 7   & -73.968  \\
    		2& -1.442      & 8   & -62.803  \\
    		3& -96.846     & 9   & -49.481  \\
    		4& -95.399     & 10  & -47.611  \\
    		5& -95.065\\
    		\hline			  
    	\end{tabular}
	\end{center}
\end{table}

\paragraph{Comparison with \textsc{MoCHi}}
We were unable to compile \textsc{MoCHi} on our machines. Instead, we attempted to use the web interface and the Dockerfile image. Comparisons were unreliable, however, due to unexpected errors, which could have been due to parsing, timeouts, or bad installation. What we noticed in some examples was that \textsc{MoCHi} took significantly longer to find bugs when we modified its assertions. For instance, checking \texttt{mult-e} with $\texttt{assert(mult m m <= mult n n)}$ took $3.667$ seconds on average, which is an increase from $0.087$ seconds for the original program (asserting $\texttt{n+1 <= mult n n}$). In contrast, at $k=1$, \textsc{BMC-2} takes $0.012$ seconds, compared to the original $0.010$ seconds.


%
%
%
\bibliographystyle{splncs04}
\bibliography{main}

\newpage
\begin{subappendices}
\renewcommand{\thesection}{\Alph{section}}
\newcommand{\wwts}[2]{We want to show:
\begin{align*}
&(#1,R,D,k)\appsigma\Downarrow(#2,R\appsigma,D\appsigma)\text{ and }\\
&\exists \sigma' \supseteq \sigma . (\sigma' \cong \phi')\text{ and }\phi' \implies (\ret = #2)\\
&\text{and }R\appsigma \subseteq R\{\sigma'\}\text{ and }D\appsigma = D\{\sigma'\}
\end{align*}}

\newcommand{\wwtss}[3]{We want to show:
\begin{align*}
&(#1,R,D,k)\appsigma\Downarrow(\chi,\hat R,S)\text{ and }\\
&\exists \sigma' \supseteq \sigma . (\sigma' \cong \phi')\text{ and }\phi' \implies (\ret = \chi)\\
&\text{and }\hat R \subseteq #2\{\sigma'\}\text{ and }S = #3\{\sigma'\}\text{, if $\chi \not\in \{\fail,\nil\}$}
\end{align*}}

\section{Proof of Lemma~\ref{lem:1}}

\textbf{Lemma~\ref{lem:1}}~~%
Given the following:
	\begin{itemize}
	\item an extended configuration $(M,R,D,k)\appsigma$ which is terminating;
	\item a translation $\llbracket M,R,C,D,\phi,k \rrbracket =(\ret,\phi',R',C',D')$;
	\item the extended assignment $\sigma \cong \phi$;
	\end{itemize}
then:
\begin{itemize}
	\item $(M,R,D,k)\appsigma \Downarrow (\chi,\hat{R},\hat S)$ and $\exists \sigma' \supseteq \sigma. (\sigma'\cong\phi')$ and $\phi' \implies (\ret=\chi)$, and
	\item if \emph{$\chi \notin \{\fail,\nil\}$}, then
          $\hat R \subseteq R'\{\sigma'\}$ and $\hat S = D'\{\sigma'\}$.
\end{itemize}

\begin{proof}
By structural induction on $M$ and by induction on the size of the derivation of the semantics of $(M,R,D,k)\appsigma$, we have the following base cases:
\begin{enumerate}
\item $k=\nil$:

We want to show:
\begin{align*}
&(M,R,D,\nil)\appsigma\Downarrow(\nil,R\appsigma,D\appsigma)\text{ and }\\
&\exists \sigma' \supseteq \sigma . (\sigma' \cong \phi')\text{ and }\phi' \implies (\ret = \nil)
\end{align*}
We choose $\sigma' = \sigma[\ret \mapsto \nil]$. Since $\phi \cong \sigma$ and $\ret$ is fresh, we have $\phi' \cong \sigma'$. Because $\phi' = (\ret = \nil) \land \phi$, we know $\phi' \implies (\ret = \nil)$.

\item $M = \fail$ and $k\neq\nil$:

We want to show:
\begin{align*}
&(\fail,R,D,k)\appsigma\Downarrow(\fail,R\appsigma,D\appsigma)\text{ and }\\
&\exists \sigma' \supseteq \sigma . (\sigma' \cong \phi')\text{ and }\phi' \implies (\ret = \fail)
\end{align*}
Similarly to the $\nil$ case, we choose $\sigma' = \sigma[\ret \mapsto \fail]$.

\item $M = v = i,(),m$ and $k\neq\nil$:

\wwts{v}{v}
We choose $\sigma' = \sigma[\ret \mapsto v]$. As before, $\phi' \cong \sigma'$ and $\phi' \implies (\ret = v)$. Additionally, since $\ret$ is fresh, it is not in $R$ or $D$. So $R\appsigma \subseteq R\{\sigma'\}$ and $D\appsigma = D\{\sigma'\}$.

\item $M = x$ and $k\neq\nil$:

\wwts{x}{\chi}
We know $x\neq\nil$ by termination of the configuration. As such, setting $\chi = \sigma(x)$, $\chi$ cannot be $\nil$. We thus choose $\sigma' = \sigma[\ret \mapsto \chi]$, which holds as previous cases. Additionally, as with case $v$, $R\appsigma \subseteq R\{\sigma'\}$ and $D\appsigma = D\{\sigma'\}$.

\item $M = {!}r$ and $k\neq\nil$:

\wwts{{!}r}{v}
Let us set $v = \sigma(D(r))$. We thus choose $\sigma' = \sigma[\ret \mapsto v]$, which holds as with previous cases. Again, $R\appsigma \subseteq R\{\sigma'\}$ and $D\appsigma = D\{\sigma'\}$.

\item $M = \lambda x.M$ and $k\neq\nil$:

We want to show:
\begin{align*}
&(\lambda x.M,R,D,k)\appsigma\Downarrow(\hat m,R[\hat m \mapsto \lambda x.M]\appsigma,D\appsigma)\text{ and }\\
&\exists \sigma' \supseteq \sigma . (\sigma' \cong \phi')\text{ and }\phi' \implies (\ret = \hat m)\\
&\text{and }R[\hat m \mapsto \lambda x.M]\appsigma \subseteq R[m \mapsto \lambda x.M]\{\sigma'\}\text{ and }D\appsigma = D\{\sigma'\}
\end{align*}
By nominal determinacy of the operational semantics, because $m$ is fresh, we can choose $\hat m$ such that $\hat m = m$. With this, we choose $\sigma' = \sigma[\ret \mapsto \hat m]$. Therefore, as before, $\phi' \cong \sigma'$ and $\phi' \implies (\ret = \hat m)$. Additionally, since $m$ and $\hat m$ are fresh, and we chose $\hat m = m$, we know that $R[\hat m \mapsto \lambda x.M] = R[m \mapsto \lambda x.M]$, so, as previously, $R[\hat m \mapsto \lambda x.M]\appsigma \subseteq R[m \mapsto \lambda x.M]\{\sigma'\}$ and $D\appsigma = D\{\sigma'\}$.
\end{enumerate}

We proceed with the inductive cases, where $k\neq\nil$:

\begin{enumerate}
\item $M = \pi_i M$:

\wwtss{\pi_i M}{R_1}{D_1}
By cases on the operational semantics:
\begin{enumerate}
\item if $M\appsigma$ does not abort, by rule $(\Downarrow_{\pi_i})$, we know that:
\[(M,R,D,k)\appsigma\Downarrow(\pair{v_1,v_2},\hat R,S)\]
Since $\sigma \cong \phi$, by the Inductive Hypothesis, we know that:
\begin{align*}
&\exists \sigma_1 \supseteq \sigma . \sigma_1 \cong \phi_1\text{ and }\phi_1 \implies (\ret_1 = \pair{v_1,v_2})\\
&\text{ and }\hat R_1 \subseteq R_1\{\sigma_1\}\text{ and }S=D_1\{\sigma_1\}
\end{align*}
We choose $\sigma' = \sigma_1[\ret \mapsto v_i]$. Since $M$ does not abort and $\pi_i\pair{v_1,v_2} = v_i$, we know $\chi = v_i$. Therefore, $\sigma' \cong \phi'$ and $\phi' \implies (\ret = \chi)$. Additionally, since $\sigma' \supseteq \sigma_1$ and $\ret$ is fresh, and because $\hat R_1 \subseteq R_1\{\sigma_1\}$ and $S = D_1\{\sigma_1\}$, we know that $\hat R_1 \subseteq R_1\{\sigma'\}$ and $S = D_1\{\sigma'\}$.

\item if $M\appsigma$ aborts, then by rule $(\not\Downarrow_{\pi_i})$, we have:
\[(M,R,D,k)\appsigma\Downarrow(\chi,R\appsigma,D\appsigma)\]
Since $\sigma \cong \phi$, we proceed by the Inductive Hypothesis as previously:
\begin{align*}
&\exists \sigma_1 \supseteq \sigma . \sigma_1 \cong \phi_1\text{ and }\phi_1 \implies (\ret = \chi)
\end{align*}
Let us now set $\chi = \sigma(\ret_1)$, where $\chi$ must be $\fail$ or $\nil$. We then choose $\sigma' = \sigma_1[\ret \mapsto \sigma(\ret_1)]$, so $\sigma' \cong \phi'$ as with the previous case.
\end{enumerate}
Case $r:=M$ is proven similarly, where $M$ evaluates to some values $v$ instead of a pair, and the whole term evaluates to $()$ instead of $v_i$.

\item $M = M_1 \oplus M_2$:

\wwtss{M_1 \oplus M_2}{R_2}{D_2}
By cases on the operational semantics:
\begin{enumerate}
\item if neither $M_1\appsigma$ or $M_2\appsigma$ abort, then by rule $(\Downarrow_\oplus)$, we have:
\begin{align*}
(M_1,R,D,k)\appsigma&\Downarrow(i_1,\hat R_1,S_1)\\
(M_2\appsigma,\hat R_1,S_1,k)&\Downarrow(i_2,\hat R_2,S_2)
\end{align*}
Since $\sigma \cong \phi$, by the Inductive Hypothesis, we know that:
\begin{align*}
&\exists \sigma_1 \supseteq \sigma . \sigma_1 \cong \phi_1\text{ and }\phi_1 \implies (\ret_1 = i_1)\\
&\text{ and }\hat R_1 \subseteq R_1\{\sigma_1\}\text{ and }S_1=D_1\{\sigma_1\}
\end{align*}
And, because $\hat R_1 \subseteq R_1\appsigmas{1}$ and $S_1 = D_1\appsigmas{1}$, by weakening the configuration, we have: 
\[(M_2\appsigma,R_1\appsigmas{1},S_1,k)\Downarrow(i_2,\hat R_2',S_2)\]
such that $\hat R_2' \supseteq \hat R_2$. Now, since $\sigma_1 \cong \phi_1$, by the Inductive Hypothesis:
\begin{align*}
&\exists \sigma_2 \supseteq \sigma_1 . \sigma_2 \cong \phi_2\text{ and }\phi_2 \implies (\ret_2 = i_2)\\
&\text{ and }\hat R_2' \subseteq R_2\{\sigma_2\}\text{ and }S_2=D_2\{\sigma_2\}
\end{align*}
We choose $\sigma' = \sigma_2[\ret \mapsto i]$ where $i = i_1 \oplus i_2$. Let us set $\chi = i$. As with earlier cases, we know $\phi' \cong \sigma'$ and $\phi' \implies (\ret = \chi)$. Additionally, we have $\hat R_2 \subseteq \hat R_2' \subseteq R_2\appsigmas{2}$ and $S_2 = D_2\appsigmas{2}$.

\item if $M_1\appsigma$ aborts, by rule $(\not \Downarrow_{\oplus_1})$, we have:
\[(M_1,R,D,k)\appsigma\Downarrow(\chi,R\appsigma, D\appsigma)\]
Since $\sigma \cong \phi$, by the Inductive Hypothesis we know:
\[\exists \sigma_1 \supseteq \sigma . \sigma_1 \cong \phi_1\text{ and }\phi_1 \implies (\ret_1 = \chi)\]
And since $\phi_1 \cong \sigma_1$, by Lemma~\ref{lem:2}, we know:
\[\exists \sigma_2 \supseteq \sigma_1 . \sigma_2 \cong \phi_2\]
Let us set $\chi = \sigma_2(\ret_1)$ where $\chi \in \{\fail,\nil\}$. We now choose $\sigma' = \sigma_2[\ret \mapsto \sigma_2(\ret_1)]$, so $\phi' \cong \sigma'$ and $\phi' \implies (\ret = \chi)$.

\item if $M_1\appsigma$ does not abort, but $M_2\appsigma$ aborts, then by rule $(\not \Downarrow_{\oplus_2})$:
\begin{align*}
(M_1,R,D,k)\appsigma&\Downarrow(i_1,\hat R_1,S_1)\\
(M_2\appsigma,\hat R_1,S_1,k)&\Downarrow(\chi,\hat R_1,S_1)
\end{align*}
This case is proven like case (a), where we choose $\sigma' = \sigma_2[\ret \mapsto \sigma_2(\ret_2)]$ instead, and $\chi = \sigma_2(\ret_2)$ must be either $\fail$ or $\nil$. We do not need to show $\hat R_2 \subseteq R_2\appsigmas{2}$ and $S_2 = D_2\appsigmas{2}$ here.
\end{enumerate}
Case $\pair{M_1,M_2}$ is proven in the same way, except we evaluate to values $v_1$ and $v_2$ instead of $i_1$ and $i_2$, and instead of $i = i_1 \oplus i_2$, we have $\pair{v_1,v_2}$.

\item $M = \text{let $x=M$ in $M'$}$:

\wwtss{\text{let $x=M$ in $M'$}}{R_2}{D_2}
By cases on the operational semantics:
\begin{enumerate}
\item if neither $M\appsigma$ or $M'\{v_1/x\}\appsigma$ abort, then by rule $(\Downarrow_{\text{let}})$:
\begin{align*}
(M,R,D,k)\appsigma&\Downarrow(v_1,\hat R_1,S_1)\\
(M'\{v_1/x\}\appsigma,\hat R_1,S_1,k)&\Downarrow(v_2,\hat R_2,S_2)
\end{align*}
Since $\sigma \cong \phi$, by the Inductive Hypothesis, we know that:
\begin{align*}
&\exists \sigma_1 \supseteq \sigma . \sigma_1 \cong \phi_1\text{ and }\phi_1 \implies (\ret_1 = v_1)\\
&\text{ and }\hat R_1 \subseteq R_1\{\sigma_1\}\text{ and }S_1=D_1\{\sigma_1\}
\end{align*}
And, because $\hat R_1 \subseteq R_1\appsigmas{1}$ and $S_1 = D_1\appsigmas{1}$, then, again, by weakening and because $\sigma_1 \cong \phi_1$ and $\phi_1 \implies (\ret_1 = v_1)$, we have: 
\[(M'\{\ret_1/x\},R_1,S_1,k)\appsigmas{1}\Downarrow(v_2,\hat R_2',S_2)\]
such that $\hat R_2' \supseteq \hat R_2$. Now, since $\sigma_1 \cong \phi_1$, by the Inductive Hypothesis:
\begin{align*}
&\exists \sigma_2 \supseteq \sigma_1 . \sigma_2 \cong \phi_2\text{ and }\phi_2 \implies (\ret_2 = v_2)\\
&\text{ and }\hat R_2' \subseteq R_2\{\sigma_2\}\text{ and }S_2=D_2\{\sigma_2\}
\end{align*}
We choose $\sigma' = \sigma_2[\ret \mapsto v_2]$. Let us set $\chi = v_2$. As earlier, we know $\phi' \cong \sigma'$ and $\phi' \implies (\ret = \chi)$. We also have $\hat R_2 \subseteq \hat R_2' \subseteq R_2\appsigmas{2}$ and $S_2 = D_2\appsigmas{2}$.

\item if $M\appsigma$ aborts, then by rule $(\not\Downarrow_{\text{let}_1})$:
\begin{align*}
(M,R,D,k)\appsigma&\Downarrow(\chi,R\appsigma,D\appsigma)
\end{align*}
This case is identical to $2.(b)$.

\item if $M\appsigma$ does not abort, but $M'\{v_1/x\}\appsigma$ does, then by rule $(\not\Downarrow_{\text{let}_2})$:
\begin{align*}
(M,R,D,k)\appsigma&\Downarrow(v_1,\hat R_1,S_1)\\
(M'\{v_1/x\}\appsigma,\hat R_1,S_1,k)&\Downarrow(\chi,\hat R_1,S_1)
\end{align*}
This case is identical to $2.(c)$.
\end{enumerate}
Case $mM$ is proven similarly, where we choose $\sigma' = \sigma_2[\ret \mapsto v_2]$ instead, and decrement $k$ upon substitution, which gives us the evaluation rule: \[(N\{v_1/x\}\appsigma{1},\hat R_1,S_1,k-1)\Downarrow(v_2,\hat R_2,S_2)\]

\item $M = \text{letrec $f=\lambda x.M$ in $M'$}$:

\wwtss{\text{letrec $f=\lambda x.M$ in $M'$}}{R'}{D'}
We choose $\sigma_1 = \sigma[f' \mapsto m]$. Since $\sigma_1 \cong (f' = m)\land\phi$, this case is directly proven by the inductive hypothesis.

\item $M = \text{if $M_b$ then $M_1$ else $M_0$}$:

\wwtss{\text{if $M_b$ then $M_1$ else $M_0$}}{R_1}{C'}
By cases on the operational semantics:
\begin{enumerate}
\item if $M_b\appsigma$ evaluates to $i$, and $M_j\appsigma$ does not abort--where $j=0$ if $i=0$, and $j=1$ otherwise--then by rule $(\Downarrow_{\text{if}})$:
\begin{align*}
(M_b,R,D,k)\appsigma&\Downarrow(i,\hat R_b,S_b)\\
(M_j\appsigma,\hat R_b,S_b,k)&\Downarrow(v,\hat R_j,S_j)
\end{align*}
Since $\sigma \cong \phi$, by the Inductive Hypothesis, we know that:
\begin{align*}
&\exists \sigma_b \supseteq \sigma . \sigma_b \cong \phi_b\text{ and }\phi_b \implies (\ret_b = j)\\
&\text{ and }\hat R_b \subseteq R_b\{\sigma_b\}\text{ and }S_b=D_b\{\sigma_b\}
\end{align*}
Now, by cases on $j$:
\begin{enumerate}
\item if $j=0$ then, since $\hat R_b \subseteq R_b\appsigmas{b}$ and $S_b = D_b\appsigmas{b}$, by weakening and because $\sigma_b \cong \phi_b$, we have: 
\[(M_0,R_b,S_b,k)\appsigmas{b}\Downarrow(v,\hat R_0',S_0)\]
such that $\hat R_0' \supseteq \hat R_0$. Since $\sigma_0 \cong \phi_0$, by the Inductive Hypothesis:
\begin{align*}
&\exists \sigma_0 \supseteq \sigma_b . \sigma_0 \cong \phi_0\text{ and }\phi_0 \implies (\ret_0 = v)\\
&\text{ and }\hat R_0' \subseteq R_0\{\sigma_0\}\text{ and }S_0=D_0\{\sigma_0\}
\end{align*}
Now, because $\phi_0 \cong \sigma_0$, by Lemma~\ref{lem:2} we have:
\[\exists \sigma_1 \supseteq \sigma_0 . \sigma_1 \cong \phi_1\]
We can then choose $\sigma' = \sigma_1[\ret \mapsto v]$. Let us set $\chi = v$. As earlier, we know $\phi' \cong \sigma'$ and $\phi' \implies (\ret = \chi)$. We also know by Lemma~\ref{lem:3} that $R_0$ must be preserved in $R_1$, which means $\hat R_0 \subseteq \hat R_0' \subseteq R_0\appsigmas{0} \subseteq R_1\appsigmas{1}$. Additionally, because $i=0$, we know $S_0 = D_0\appsigmas{0}$.

\item if $j=1$ then, since $\phi_b \cong \sigma_b$, we have by Lemma~\ref{lem:2}:
\[\exists \sigma_0 \supseteq \sigma_b . \sigma_0 \cong \phi_0\]
Now, by Lemma~\ref{lem:3}, we know $R_b$ must be preserved in $R_0$, so we have $\hat R_b \subseteq R_b \subseteq R_0$. This gives us:
\[(M_1,R_0,S_b,k)\appsigmas{0}\Downarrow(v,\hat R_1',S_1)\]
such that $\hat R_1' \supseteq R_1$. Thus, because $\sigma_0 \cong \phi_0$, we know by the Inductive Hypothesis that:
\begin{align*}
&\exists \sigma_1 \supseteq \sigma_0 . \sigma_1 \cong \phi_1\text{ and }\phi_1 \implies (\ret_1 = v)\\
&\text{ and }\hat R_1' \subseteq R_1\{\sigma_1\}\text{ and }S_1=D_1\{\sigma_1\}
\end{align*}
We now choose $\sigma' = \sigma_1[\ret \mapsto v]$. Let us set $\chi = v$. As earlier, we know $\phi' \cong \sigma'$ and $\phi' \implies (\ret = \chi)$. We also know $\hat R_1' \subseteq R_1\appsigmas{1}$ and because $i\neq 0$, we know $S_1 = D_1\appsigmas{1}$.
\end{enumerate}

\item if $M\appsigma$ aborts, then by rule $(\not\Downarrow_{\text{if}_1})$:
\begin{align*}
(M_b,R,D,k)\appsigma&\Downarrow(\chi,R\appsigma,D\appsigma)
\end{align*}
Since $\sigma \cong \phi$, by the Inductive Hypothesis we know:
\[\exists \sigma_b \supseteq \sigma . \sigma_b \cong \phi_b\text{ and }\phi_b \implies (\ret_b = \chi)\]
Then since $\phi_b \cong \sigma_b$, by Lemma~\ref{lem:2}, we know:
\[\exists \sigma_0 \supseteq \sigma_b . \sigma_0 \cong \phi_0\]
And because $\phi_0 \cong \sigma_0$, by Lemma~\ref{lem:2} again, we know:
\[\exists \sigma_1 \supseteq \sigma_0 . \sigma_1 \cong \phi_1\]
Let us then set $\chi = \sigma_1(\ret_b)$ where $\chi \in \{\fail,\nil\}$. We choose $\sigma' = \sigma_1[\ret \mapsto \sigma_1(\ret_b)]$, so $\phi' \cong \sigma'$ and $\phi' \implies (\ret = \chi)$.

\item if $M_b\appsigma$ evaluates to $i$, but $M_j\appsigma$ aborts--where $j=0$ if $i=0$, and $j=1$ otherwise--then by rule $(\not\Downarrow_{\text{if}_2})$:
\begin{align*}
(M_b,R,D,k)\appsigma&\Downarrow(i,\hat R_b,S_b)\\
(M_j\appsigma,\hat R_b,S_b,k)&\Downarrow(\chi,\hat R_b,S_b)
\end{align*}
This case is proven like case (a), where we choose $\sigma' = \sigma_1[\ret \mapsto \sigma_1(\ret_j)]$ instead, and $\chi = \sigma_1(\ret_j)$ must be either $\fail$ or $\nil$. Additionally, because the configuration aborts, we are not required to show conditions involving the repository and store.
\end{enumerate}

\item $M = x M$:

\wwtss{x M}{R_n}{C_n'}
Let us set $\sigma(x) = m_j$ for some $j \in \{1,\dots,n\}$. By cases on the operational semantics:
\begin{enumerate}
\item if neither $M\appsigma$ or $N_j\{v/y_j\}\appsigma$ abort, then by rule $(\Downarrow_{\text{@}})$:
\begin{align*}
(M,R,D,k)\appsigma&\Downarrow(v_0,\hat R_0,S_0)\\
(N_j\{v/y_j\}\appsigma,\hat R_0,S_0,k-1)&\Downarrow(v_j,\hat R_j,S_j)
\end{align*}
Since $\sigma \cong \phi$, by the Inductive Hypothesis, we know that:
\begin{align*}
&\exists \sigma_0 \supseteq \sigma . \sigma_0 \cong \phi_0\text{ and }\phi_0 \implies (\ret_0 = v_0)\\
&\text{ and }\hat R_0 \subseteq R_0\{\sigma_0\}\text{ and }S_0=D_0\{\sigma_0\}
\end{align*}
Now, for every $i \in \{1,\dots,j-1\}$, we can consecutively apply Lemma~\ref{lem:2}, starting from $\sigma_0 \cong \phi_0$, to obtain:
\[\exists \sigma_{j-1} \supseteq \dots \supseteq \sigma_0 . \sigma_{j-1} \cong \phi_{j-1}\]
Then, because $\hat R_0 \subseteq R_0\appsigmas{0} \subseteq \dots \subseteq R_{j-1}\appsigmas{{j-1}}$ and $S_0 = D_0\appsigmas{0} = D_0\appsigmas{{j-1}}$, by weakening and because $\sigma_{j-1} \cong \phi_{j-1}$, we have: 
\[(N_j\{v/y_j\}\appsigmas{{j-1}},R_{j-1}\appsigmas{{j-1}},S_0,k)\Downarrow(v_j,\hat R_j',S_j)\]
such that $\hat R_j' \supseteq \hat R_j$. So, because $\sigma_{j-1} \cong \phi_{j-1}$, by the Inductive Hypothesis:
\begin{align*}
&\exists \sigma_j \supseteq \sigma_{j-1} . \sigma_j \cong \phi_j\text{ and }\phi_j \implies (\ret_j = v_j)\\
&\text{ and }\hat R_j' \subseteq R_j\{\sigma_0\}\text{ and }S_j=D_j\{\sigma_j\}
\end{align*}
Now, because $\phi_j \cong \sigma_j$, again by consecutively applying Lemma~\ref{lem:2}, we have:
\[\exists \sigma_n \supseteq ... \sigma_j . \sigma_n \cong \phi_n\]
We can then choose $\sigma' = \sigma_n[\ret \mapsto v_j]$. Let us set $\chi = v_j$. Again, we know $\phi' \cong \sigma'$ and $\phi' \implies (\ret = \chi)$. We also know by Lemma~\ref{lem:3} that $R_j$ must be preserved in $R_n$, which means $\hat R_j \subseteq \hat R_j' \subseteq R_j\appsigmas{j} \subseteq R_n\appsigmas{n}$. Additionally, we know $S_j = D_j\appsigmas{n}$.

\item if $M\appsigma$ aborts, then by rule $(\not\Downarrow_{\text{@}_1})$:
\begin{align*}
(M_b,R,D,k)\appsigma&\Downarrow(\chi,R\appsigma,D\appsigma)
\end{align*}
Since $\sigma \cong \phi$, by the Inductive Hypothesis we know:
\[\exists \sigma_0 \supseteq \sigma . \sigma_0 \cong \phi_0\text{ and }\phi_0 \implies (\ret_0 = \chi)\]
Then since $\phi_0 \cong \sigma_0$, as before, by applying Lemma~\ref{lem:2} consecutively, we know:
\[\exists \sigma_n \supseteq ... \supseteq \sigma_0 . \sigma_n \cong \phi_n\]
Let us then set $\chi = \sigma_n(\ret_0)$ where $\chi \in \{\fail,\nil\}$. We choose $\sigma' = \sigma_n[\ret \mapsto \sigma_n(\ret_0)]$, so $\phi' \cong \sigma'$ and $\phi' \implies (\ret = \chi)$.

\item if $M\appsigma$ does not abort, but $N_j\{v/y_j\}\appsigma$ aborts, then by rule $(\not\Downarrow_{\text{@}_2})$ we have:
\begin{align*}
(M,R,D,k)\appsigma&\Downarrow(v,\hat R_0,S_0)\\
(N_j\{v/y_j\}\appsigma,\hat R_0,S_0,k-1)&\Downarrow(\chi,\hat R_0,S_0)
\end{align*}
This case is proven like case (a), where we choose $\sigma' = \sigma_n[\ret \mapsto \sigma_n(\ret_j)]$, and $\chi = \sigma_n(\ret_j)$ must be either $\fail$ or $\nil$. Additionally, because the configuration aborts, we do not need to prove conditions for the repository and store.
\end{enumerate}
\end{enumerate}
\end{proof}
\newcommand{\wwtz}[1]{We want to show:
\[\exists \sigma' \supseteq \sigma . \sigma' \cong (#1)\]}

\begin{lemma}[Uniqueness of the translation]\label{lem:2}
Given an assignment $\sigma$ and $\phi$ where $\sigma \cong \phi$, and a translation $\sem{M,R,C,D,\phi,k}=(\ret,\phi',R',C',D')$, we know there exists some $\sigma' \supseteq \sigma$ such that $\sigma' \cong \phi'$.
\end{lemma}
\begin{proof}
Assuming $\sigma \cong \phi$, by induction on $k$ and then by structural induction on $M$, we have the base cases:
\begin{enumerate}
\item $k=\nil$:

\wwtz{(\ret = \nil) \land \phi}
We choose $\sigma' = \sigma[\ret \mapsto \nil]$. Since $\sigma \cong \phi$, and the only fresh name $\ret$ maps to $\nil$ in $\sigma'$, we know $\sigma' \cong ((\ret = \nil) \land \phi)$.

\item $M = \fail$ and $k\neq\nil$:

\wwtz{(\ret = \fail) \land \phi}
We choose $\sigma' = \sigma[\ret \mapsto \fail]$, so $\sigma' \cong ((\ret = \fail) \land \phi)$. Similarly for the remaining base cases:
\begin{enumerate}
\item for $M = v = i,()$, choose $\sigma' = \sigma[\ret \mapsto v]$, so $\sigma' \cong ((\ret = v) \land \phi)$.
\item for $M = m$, choose $\sigma' = \sigma[\ret \mapsto m]$, so $\sigma' \cong ((\ret = m) \land \phi)$.
\item for $M = x$, choose $\sigma' = \sigma[\ret \mapsto x]$, so $\sigma' \cong ((\ret = x) \land \phi)$.
\item for $M = {!}r$, choose $\sigma' = \sigma[\ret \mapsto \sigma(D(r))]$, so $\sigma' \cong ((\ret = D(r)) \land \phi)$.
\item for $M = \lambda x.M$, choose $\sigma' = \sigma[\ret \mapsto m]$, so $\sigma' \cong ((\ret = m) \land \phi)$.
\end{enumerate}
\end{enumerate}
With base cases done, we recall predicate formula $F$:
\begin{align*}
F~a~b~P = &((a=\fail)\implies (b=\fail))\ \land\\
          &((a=\nil) \implies (b=\nil))\  \land\\
          &((a=\fail) \lor (a=\nil) \lor P)
\end{align*}
We then have the following inductive cases:
\begin{enumerate}
\item $M = \pi_i M$:

\wwtz{(F~\ret_1~\ret(\ret=\pi_i\ret_1)) \land \phi_1}
We have $\sigma \cong \phi$, so by the Inductive Hypothesis:
\[\exists \sigma_1 \supseteq \sigma . \sigma_1 \cong \phi_1\]
Let us set $\chi = \sigma_1(\ret_1)$. By cases on $\chi$:
\begin{enumerate}
\item if $\chi = \pair{v_1,v_2}$, we choose $\sigma' = \sigma_1[\ret \mapsto v_i]$.
\item if $\chi \in \{\fail,\nil\}$, we choose $\sigma' = \sigma_1[\ret \mapsto \chi]$.
\end{enumerate}

\item $M = r:=M$:

\wwtz{(F~\ret_1~\ret(\ret=\pi_i\ret_1)) \land \phi_1}
We have $\sigma \cong \phi$, so by the Inductive Hypothesis:
\[\exists \sigma_1 \supseteq \sigma . \sigma_1 \cong \phi_1\]
Let us set $\chi = \sigma_1(\ret_1)$. By cases on $\chi$:
\begin{enumerate}
\item if $\chi = v$, we choose $\sigma' = \sigma_1[\ret \mapsto (), D(r) \mapsto v]$.

\item if $\chi \in \{\fail,\nil\}$, we choose $\sigma' = \sigma_1[\ret \mapsto \chi]$. We then know $\sigma'$ uniquely satisfies $\phi'$ up to the disjoint succeeding clause of $F$, which contains $D(r)$. Since the disjunction is trivially true, we are allowed to extend $\sigma'$ to map $D(r)$ to an arbitrary value, e.g. $\sigma''= \sigma'[D(r)\mapsto \chi]$. Because there exists a $\sigma'' \supseteq \sigma'$ that satisfies $\phi'$, and $\phi'$ implies $\sigma'$, we have $\sigma' \cong \phi'$.
\end{enumerate}

\item $M = M_1 \oplus M_2$:

\wwtz{(F~\ret_1~\ret(F~\ret_2~\ret(\ret=\ret_1\oplus\ret_2))) \land \phi_2}
We have $\sigma \cong \phi$, so by the Inductive Hypothesis:
\[\exists \sigma_1 \supseteq \sigma . \sigma_1 \cong \phi_1\]
Since $\sigma_1 \cong \phi_1$, by the Inductive Hypothesis:
\[\exists \sigma_2 \supseteq \sigma_1 . \sigma_2 \cong \phi_2\]
Let us set $\chi_1 = \sigma_2(\ret_1)$ and $\chi_2 = \sigma_2(\ret_2)$. By cases on $\chi_1,\chi_2$:
\begin{enumerate}
\item if $\chi_1 = v_1$ and $\chi_2 = v_2$, we choose $\sigma' = \sigma_2[\ret \mapsto v]$, where $v = v_1 \oplus v_2$.

\item if $\chi_1 \in \{\fail,\nil\}$, we choose $\sigma' = \sigma_2[\ret \mapsto \chi_1]$.

\item if $\chi_1 = v_1$ and $\chi_2 \in \{\fail,\nil\}$, we choose $\sigma' = \sigma_2[\ret \mapsto \chi_2]$.
\end{enumerate}
A similar proof applies to cases:
\begin{enumerate}
\item $\pair{M_1,M_2}$, where $\sigma' = \sigma_2[\ret \mapsto \pair{v_1,v_2}]$ is chosen instead of $v$ in (a).

\item `$\text{let $x=M$ in $M'$}$' and `$mM$', where $\sigma' = \sigma_2[\ret \mapsto v_2]$ is chosen instead of $v$ in (a).
\end{enumerate}

\item $M = \text{letrec $f=\lambda x.M$ in $M'$}$:

\wwtz{\phi'}
where $\phi'$ is the result of translating $M'\{f'/f\}$.

Choose $\sigma_1 = \sigma[f \mapsto m]$ such that $\sigma_1 \cong (\phi \land (f'=m))$. This case is then directly proven by the inductive hypothesis.

\item $M = \text{if $M_b$ then $M_1$ else $M_0$}$:

\wwtz{(F~\ret_b~\ret(\psi_0 \land \psi_1)) \land \phi_1}
We have $\sigma \cong \phi$, so by the Inductive Hypothesis:
\[\exists \sigma_b \supseteq \sigma . \sigma_b \cong \phi_b\]
Since $\sigma_b \cong \phi_b$, by the Inductive Hypothesis:
\[\exists \sigma_0 \supseteq \sigma_b . \sigma_0 \cong \phi_0\]
Since $\sigma_0 \cong \phi_0$, by the Inductive Hypothesis:
\[\exists \sigma_1 \supseteq \sigma_0 . \sigma_1 \cong \phi_1\]
Let us set $\chi_b = \sigma_1(\ret_b)$, $\chi_0 = \sigma_1(\ret_0)$ and $\chi_1 = \sigma_1(\ret_1)$. By cases on $\chi_b$:
\begin{enumerate}
\item if $\chi_b = i$, we choose $\sigma' = \sigma_1[\ret \mapsto \chi_j, C'(r) \mapsto D_j(r)]$ for every $r \in \varPi$, where $j = 0$ if $i=0$, else $j=1$.

\item if $\chi_b \in \{\fail,\nil\}$, we choose $\sigma' = \sigma_1[\ret \mapsto \chi_b]$. We know $\sigma'$ uniquely satisfies $\phi'$ up to the disjoint succeeding clauses of $F$ involving $\psi_0$ and $\psi_1$, which contain $C'(r)$ for all $r \in \varPi$. Since the disjunction is trivially true, we are allowed to extend $\sigma'$ to map $C'(r)$ to an arbitrary value, e.g. $\sigma''= \sigma'[C(r)\mapsto \chi_b]$ for each $r \in \varPi$. Because there exists a $\sigma'' \supseteq \sigma'$ that satisfies $\phi'$, and $\phi'$ implies $\sigma'$, we have $\sigma' \cong \phi'$.
\end{enumerate}

\item $M = xM$:

\wwtz{(F~\ret_0~\ret\psi) \land \phi_n}
We have $\sigma \cong \phi$, so by the Inductive Hypothesis:
\[\exists \sigma_0 \supseteq \sigma . \sigma_0 \cong \phi_0\]
Since $\sigma_0 \cong \phi_0$, by the Inductive Hypothesis applied consecutively:
\[\exists \sigma_n \supseteq \dots \supseteq \sigma_0 . (\sigma_n \cong \phi_n) \land ... \land (\sigma_{0} \cong \phi_{0})\]
Let us set $\chi_0,\dots,\chi_n = \sigma_n(\ret_0),\dots,\sigma_n(\ret_n)$. Let us also set $\sigma(x) = m_j$ where $j \in \{1,\dots,n\}$. By cases on $\chi_0$:
\begin{enumerate}
\item if $\chi_0 = v$, we choose $\sigma' = \sigma_n[\ret \mapsto \chi_j, C'(r) \mapsto D_j(r)]$ for each $r \in \varPi$.

\item if $\chi_0 \in \{\fail,\nil\}$, we choose $\sigma' = \sigma_n[\ret \mapsto \chi_0]$. We know $\sigma'$ uniquely satisfies $\phi'$ up to the disjoint succeeding clauses of $F$ involving $\psi$, which contains $C'(r)$ for all $r \in \varPi$. Since the disjunction is trivially true, we are allowed to extend $\sigma'$ to map $C'(r)$ to an arbitrary value, e.g. $\sigma''= \sigma'[C(r)\mapsto \chi_0]$ for each $r \in \varPi$. Because there exists a $\sigma'' \supseteq \sigma'$ that satisfies $\phi'$, and $\phi'$ implies $\sigma'$, we have $\sigma' \cong \phi'$.
\end{enumerate}

\end{enumerate}
\end{proof}
\begin{lemma}[Preservation of the repository]\label{lem:3}
Given a translation $\sem{M,R,C,D,\phi,k}=(\ret,\phi',R',C',D')$, we know the input repository must be preserved; i.e. $R' \supseteq R$.
\end{lemma}
\begin{proof}
By inspection of the translation rules.
\end{proof}
\begin{lemma}[Propagation of preconditions]\label{lem:4}
Given a translation $\sem{M,R,C,D,\phi,k}=(\ret,\phi',R',C',D')$, we know that preconditions $\phi$ must be propagated and included in $\phi'$; i.e. $\phi'=\psi\land\phi$ where $\sem{M,R,C,D,\top,k}=(\ret,\psi,R',C',D')$
\end{lemma}
\begin{proof}
By inspection of the translation rules.
\end{proof}
\section{Optimised BMC translation}\label{apx:full}

\noindent
\textbf{Base Cases:}
\begin{itemize}
	\item $\llbracket M,R,C,D,pt,\phi,\nil\rrbracket = (\ret,(\ret=\nil)\land \phi,R,C,D,\varnothing,pt)$
	\item $\llbracket \fail,R,C,D,pt,\phi,k\rrbracket = (\ret,(\ret=\fail)\land \phi,R,C,D,\varnothing,pt)$
	\item $\begin{aligned}[t]
    &\llbracket v,R,C,D,pt,\phi,k\rrbracket = (\ret,(\ret=v)\land \phi,R,C,D,\varnothing,pt) \text{ where $v = i,()$}
    \end{aligned}$ 
    \item $\begin{aligned}[t]
    &\llbracket m,R,C,D,pt,\phi,k\rrbracket = (\ret,(\ret=m)\land \phi,R,C,D,\{m\},pt)
    \end{aligned}$ 
    \item $\begin{aligned}[t]
    &\llbracket x,R,C,D,pt,\phi,k\rrbracket = (\ret,(\ret=x)\land \phi,R,C,D,pt(x),pt)
    \end{aligned}$ 
	\item $\llbracket {!}r,R,C,D,pt,\phi,k\rrbracket = (\ret,(\ret=D(r))\land \phi,R,C,D,pt(r),pt)$
	\item $\begin{aligned}[t]
    &\llbracket \lambda x.M,R,C,D,pt,\phi,k\rrbracket = (\ret,(\ret=m)\land \phi,R',C,D,\{m\},pt)\\
    &\quad\text{ where $R'=R[m \mapsto \lambda x.M]$ and $m$ fresh}
    \end{aligned}$
\end{itemize}

\noindent
\textbf{Inductive Cases:}
\begin{itemize}
	\item $
    	\begin{aligned}[t]
    		& \llbracket\pi_i\,M,R,C,D,pt,\phi,k\rrbracket =\\
    		&\quad \begin{aligned}[t]
    			& \letin{(\ret_1,\phi_1,R_1,C_1,D_1,A_1,pt_1) = \llbracket M,R,C,D,pt,\phi,k\rrbracket}\\
    			& (\ret,(F~\ret_1~\ret~(\ret=\pi_i\,ret_1))\land \phi_1,R_1,C_1,D_1,\pi_i\,A_1,pt_1)
    		\end{aligned}
    	\end{aligned}$
	\item $
    	\begin{aligned}[t]
    		& \llbracket r:=M,R,C,D,pt,\phi,k\rrbracket =\\
    		&\quad \begin{aligned}[t]
    			& \letin{(\ret_1,\phi_1,R_1,C_1,D_1,A_1,pt_1) = \llbracket M,R,C,D,pt,\phi,k\rrbracket}\\
    			& \letin{C_1' = C_1[r]}\
    			\letin{D_1' = D_1[r\mapsto C_1'(r)]}\\
    			& (\ret,(F~\ret_1~\ret~((\ret={()}) \land (D_1'(r)=\ret_1)))\land \phi_1,R_1,C_1',D_1',\varnothing,pt_1[r \mapsto A_1])
    		\end{aligned}
    	\end{aligned}$
	\item $
    	\begin{aligned}[t]
    		& \llbracket M_1 \oplus M_2,R,C,D,pt,\phi,k\rrbracket =\\
    		&\quad \begin{aligned}[t]
    			& \letin{(\ret_1,\phi_1,R_1,C_1,D_1,A_1,pt_1) = \llbracket M,R,C,D,pt,\phi,k\rrbracket}\\
    			& \letin{(\ret_2,\phi_2,R_2,C_2,D_2,A_2,pt_2) = 
    			 \llbracket M_2,R_1,C_1,D_1,pt_1,\phi_1,k\rrbracket}\\
    			& (\ret,(F~\ret_1~\ret~(F~\ret_2~\ret~(\ret=\ret_1 \oplus \ret_2))) \land \phi_2,R_2,C_2,D_2,\varnothing,pt_2)
    		\end{aligned}
    	\end{aligned}$
	\item $
    	\begin{aligned}[t]
    	& \llbracket \pair{M_1,M_2},R,C,D,pt,\phi,k\rrbracket =\\
    		&\quad \begin{aligned}[t]
    			& \letin{(\ret_1,\phi_1,R_1,C_1,D_1,A_1,pt_1) = \llbracket M,R,C,D,pt,\phi,k\rrbracket}\\
    			& \letin{(\ret_2,\phi_2,R_2,C_2,D_2,A_2,pt_2) = 
    			 \llbracket M_2,R_1,C_1,D_1,pt_1,\phi_1,k\rrbracket}\\
    			& (\ret,(F~\ret_1~\ret~(F~\ret_2~\ret~(\ret= \pair{\ret_1,\ret_2}))) \land \phi_2,R_2,C_2,D_2,\pair{A_1,A_2},pt_2)
    		\end{aligned}
    	\end{aligned}$
	\item $
    	\begin{aligned}[t]
    	& \llbracket \texttt{let $x=M$ in $M'$},R,C,D,pt,\phi,k\rrbracket =\\
    		&\quad \begin{aligned}[t]
    			& \letin{(\ret_1,\phi_1,R_1,C_1,D_1,A_1,pt_1) = \llbracket M,R,C,D,pt,\phi,k\rrbracket}\\
    			& \letin{(\ret_2,\phi_2,R_2,C_2,D_2,A_2,pt_2) =\llbracket M'\{\ret_1/x\},R_1,C_1,D_1,pt_1[\ret_1 \mapsto A_1],\phi_1,k\rrbracket}\\
    			& (\ret,(F~\ret_1~\ret~(F~\ret_2~\ret~(\ret=\ret_2)))\land \phi_2,R_2,C_2,D_2,A_2,pt_2)
    		\end{aligned}
    	\end{aligned}$
    \item $
    	\begin{aligned}[t]
    	& \llbracket \texttt{letrec $f=\lambda x.M$ in $M'$},R,C,D,pt,\phi,k\rrbracket =\\
    		&\quad \begin{aligned}[t]
    			& \letin{\text{$m,f'$ be fresh}}\\
                & \letin{R' = R[m \mapsto \lambda x.M\{f'/f\}]}\\
    			& \llbracket M'\{f'/f\},R',C,D,pt[f' \mapsto \{m\}],\phi \land (f' = m),k\rrbracket
    		\end{aligned}
    	\end{aligned}$
	\item $
    	\begin{aligned}[t]
    		& \llbracket m\,M,R,C,D,pt,\phi,k\rrbracket =\\
    		&\quad \begin{aligned}[t]
    			& \letin{(\ret_1,\phi_1,R_1,C_1,D_1,A_1,pt_1) = \llbracket M,R,C,D,pt,\phi,k\rrbracket}\\
    			& \letin{R(m)\text{ be }\lambda x.N}\\
    			& \letin{(\ret_2,\phi_2,R_2,C_2,D_2) = \llbracket N\{\ret_1/x\},R_1,C_1,D_1,pt_1[\ret_1 \mapsto A_1],\phi_1,k-1\rrbracket}\\
    			& (\ret,(F~\ret_1~\ret~(F~\ret_2~\ret~(\ret=\ret_2)))\land \phi_2,R_2,C_2,D_2,A_2,pt_2)
    		\end{aligned}
    	\end{aligned}$
	\item $
    	\begin{aligned}[t]
    		& \llbracket \texttt{if $M_b$ then $M_1$ else $M_0$},R,C,D,pt,\phi,k\rrbracket =\\
    		&\quad \begin{aligned}[t]
    			& \letin{(\ret_b,\phi_b,R_b,C_b,D_b,A_b,pt_b) = \llbracket M_b,R,C,D,pt,\phi,k\rrbracket}\\
    			& \letin{(\ret_0,\phi_0,R_0,C_0,D_0,A_0,pt_0) = 
    			 \llbracket 	M_0,R_b,C_b,D_b,pt_b,\phi_b,k\rrbracket}\\
    			& \letin{(\ret_1,\phi_1,R_1,C_1,D_1,A_1,pt_1) = 
    			 \llbracket M_1,R_0,C_0,D_b,pt_b,\phi_0,k\rrbracket}\\
    			& \letin{C' = C_1[r_1]\cdots[r_n]\ (\varPi=\{r_1,\dots,r_n\})}\\
    			& \letin{\psi_0 = (\ret_b = 0) \implies (F~\ret_0~\ret~((\ret=\ret_0) \land\!\! \bigwedge_{r \in \varPi}\!\! (C'(r)=D_1(r))))}\\
                          & \letin{\psi_1 = (\ret_b \neq 0) \implies
                            (F~\ret_1~\ret~((\ret=\ret_1) \land\!\!\bigwedge_{r \in \varPi}\!\! (C'(r)=D_1(r))))}\\
    			& (\ret,(F~\ret_b~\ret~(\psi_0 \land \psi_1))\land \phi_1,R_1,C',C',A_0 \cup A_1,merge(pt_0,pt_1))
    		\end{aligned}
    	\end{aligned}$
	\item $
    	\begin{aligned}[t]
    		& \llbracket x^\theta\,M,R,C,D,pt,\phi,k\rrbracket =\\
    		&\quad \begin{aligned}[t]
                      & \letin{(\ret_0,\phi_0,R_0,C_0,D_0,A_0,pt_0) = \llbracket M,R,C,D,pt,\phi,k\rrbracket}\\
                      &\text{if }pt(x) =\emptyset\text{ then }(\ret,(\ret=\nil)\land\phi_0,R_0,C_0,D_0,\varnothing,pt_0)\\
    &\text{else }
    			\letin{pt(x)\text{ be }\{m_1,...,m_n\}}\\
    			&\text{for each }i \in \{1,...,n\}:\\
    			&\quad
    			\begin{aligned}[t]
    				& \letin{R(m_i)\text{ be } \lambda y_i.N}\\
    				& \letin{(\ret_i,\phi_i,R_i,C_i,D_i,A_i,pt_i) = \llbracket N_i\{\ret_0/y_i\},R_{i-1},
    				C_{i-1},D_0,pt_0,\phi_{i-1},k-1\rrbracket}\\
    			\end{aligned}\\
    			& \letin{C_n' = C_n[r_1]\cdots[r_j]\ (\varPi=\{r_1,\dots,r_j\})}\\
    			& \letin{\psi = \bigwedge_{i=1}^{n}
    				\left(\vcenter{\hbox{$\displaystyle
    				\begin{aligned}[t]
    					&(x=m_i) \implies ((F~\ret_i~\ret~(\ret=\ret_i)) \land\bigwedge_{r \in \varPi}{(C_n'(r)=D_i(r))})
    				\end{aligned}$}}\right)}\\
    			& (\ret,(F~\ret_0~\ret~\psi)\land \phi_n,R_n,C_n',C_n',A_1 \cup \dots\cup A_n,merge(pt_1, \dots, pt_n))
    			\end{aligned}
    	\end{aligned}$
\end{itemize}


\end{subappendices}
\end{document}